\documentclass[letterpaper, 11pt, darftcls]{IEEEtran}      

\onecolumn

\usepackage{macros}

\author{
    \IEEEauthorblockN{
Alon~Kipnis and   
    John~C.~Duchi}

\thanks{
Alon Kipnis is with the School of Computer Science at Reichman University, Herzliya, Israel (alon.kipnis@idc.ac.il).} 
\thanks{J. Duchi is with the Department of Statistics and the Department of Electrical Engineering at Stanford University, Stanford, CA, 94035 (jduchi@stanford.edu).}
\thanks{Copyright (c) 2017 IEEE. Personal use of this material is permitted.  However, permission to use this material for any other purposes must be obtained from the IEEE by sending a request to pubs-permissions@ieee.org.}
\thanks{This paper was presented in part at the 55th Annual Allerton Conference on Communication, Control, and Computing (Allerton) \cite{KipnisAllerton2017}. }
\thanks{
The work of A. Kipnis was supported in part by funding from the NSF under Grant No.~DMS-1418362 and DMS-1407813, and by a fellowship from the Koret Foundation.}
}

\title{\LARGE \bf Mean Estimation from One-Bit Measurements}

\begin{document}
\graphicspath{{./Figs/}}
\maketitle

\begin{abstract}
  We consider the problem of estimating the mean of a symmetric log-concave
  distribution under the constraint that only a single bit per sample
  from this distribution is available to the estimator. We study the mean
  squared error as a function of the
  sample size (and hence the number of bits).
  We consider three settings: first, a centralized setting, where
  an encoder may release $n$ bits given a sample of size $n$, and for
  which there is no asymptotic penalty for quantization; second, an adaptive
  setting  in which each bit is a function of the current
  observation and previously recorded bits, where we show that
  the optimal relative
  efficiency compared to the sample mean is precisely
  the efficiency of the median; lastly, we show that in
  a distributed setting where each bit is only a function
  of a local sample, no estimator can achieve optimal efficiency
  uniformly over the parameter space.
  We additionally complement our results in the adaptive setting
  by showing that \emph{one} round of adaptivity is sufficient
  to achieve optimal mean-square error.
\end{abstract}



\section{Introduction}
\label{sec:Intro}

We consider estimation of parameters from data collected by multiple units
under communication constraints between the units.  Such scenarios arise in
sensor arrays, where sensor motes collect information, which they transmit
to a central estimation unit~\cite{LesserOrTa03,LiWoHuSa02}. More generally,
communication is substantially more expensive than computation in modern
computing infrastructure~\cite{FullerMi11}.  It is thus of interest to
understand the extent to which communication constraints induce fundamental
accuracy and efficiency limits in parametric estimation problems.
\begin{figure*}
  \begin{center}
\begin{tikzpicture}[node distance=2cm,auto]

  \node at (0,0) (source) {$X_1$} ;
 \node[below of = source, node distance = 1.5cm] (source2) {$X_2$};
\node[below of = source2, node distance = 2.1cm] (source3) {$X_n$};

\node[int1, right of = source2, node distance = 1.2cm] (enc2) {$\enc$};  

\draw[->,line width = 2pt] (source2) -- (enc2); 
\draw[->,line width = 2pt] (source) -| (enc2); 
\draw[->,line width = 2pt] (source3) -| (enc2); 

\node[below of = source2, node distance = 0.5cm] {$\vdots$};

\draw[->,line width = 2pt] (source2) -- (enc2); 
\node[int1, right of = enc2, node distance = 2.7cm ] (est) {$\est$};

\draw[->,line width = 2pt] (enc2) -- node[above, xshift = 0cm] (mes2) {\small $B_1,\ldots,B_n$} (est);   

\node[right of = est, node distance = 1.2cm] (dest) {\small ${\theta}_n$};
\draw[->, line width=1pt] (est) -- (dest);
\node at (2,-4.5) {(i) Centralized};
\end{tikzpicture}
\begin{tikzpicture}[node distance=2cm,auto]
  \node at (0,0) (source) {\small $X_1$} ;
  \node[int1, right of = source, node distance = 1.2cm] (enc1) {$\enc$};  
\draw[->,line width = 2pt] (source) -- (enc1); 

 \node[below of = source, node distance = 1.5cm] (source2) {\small $X_2$};
\node[int1, right of = source2, node distance = 1.2cm] (enc2) {$\enc$};  
\draw[->,line width = 2pt] (source2) -- (enc2); 

\node[below of = source2, node distance = 2.1cm] (source3) {\small $X_n$};
\node[int1, right of = source3, node distance = 1.2cm] (enc3) {$\enc$};  
\draw[->,line width = 2pt] (source3) -- (enc3); 

\node[below of = source2, node distance = 0.5cm] {$\vdots$};

\node[int1, right of = enc2, node distance = 3cm ] (est) {$\est$};
\draw[->,line width = 2pt] (enc1) -| node[above, xshift = -1cm] (mes1) {$B_1$} (est);   
\draw[->,line width = 2pt] (enc2) -- node[above, xshift = 0cm] (mes2) {$B_2$} (est);   

\draw[->] (enc1)+(0.7,0) -- +(0.7,-0.7) -| (enc2);

\draw[->,line width = 2pt] (enc3) -| (est);   


\draw[->,line width = 2pt] (enc3) -| node[above, xshift = -1cm]  {$B_n$} (est);   


\node[below right = 2.3cm and 0.7cm of 
enc1.center, yshift = 0.cm] (end) {$B_1,\ldots,B_{n-1}$};

\draw[->] (enc1)+(0.7,0) -- (end.west) -| (enc3);

\node[right of = est, node distance = 1.2cm] (dest) {${\theta}_n$};
\draw[->, line width=1pt] (est) -- (dest);
\node at (2,-4.5) {(ii) Adaptive};
\end{tikzpicture}
\begin{tikzpicture}[node distance=2cm,auto]
  \node at (0,0) (source) {$X_1$} ;
  \node[int1, right of = source, node distance = 1.2cm] (enc1) {$\enc$};  
\draw[->,line width = 2pt] (source) -- (enc1); 

 \node[below of = source, node distance = 1.5cm] (source2) {$X_2$};
\node[int1, right of = source2, node distance = 1.2cm] (enc2) {$\enc$};  
\draw[->,line width = 2pt] (source2) -- (enc2); 

\node[below of = source2, node distance = 2.1cm] (source3) {$X_n$};
\node[int1, right of = source3, node distance = 1.2cm] (enc3) {$\enc$};  
\draw[->,line width = 2pt] (source3) -- (enc3); 

\node[below of = source2, node distance = 0.5cm] {$\vdots$};

\node[int1, right of = enc2, node distance = 2.1cm ] (est) {$\est$};
\draw[->,line width = 2pt] (enc1) -| node[above, xshift = -1cm] (mes1) {$B_1$} (est);   

\draw[->,line width = 2pt] (enc2) -- node[above, xshift = 0cm] (mes2) {$B_2$} (est);   

\draw[->,line width = 2pt] (enc3) -| (est);   

\draw[->,line width = 2pt] (enc3) -| node[above, xshift = -1cm]  {$B_n$} (est);   

\node[right of = est, node distance = 1.2cm] (dest) {${\theta}_n$};
\draw[->, line width=1pt] (est) -- (dest);
\node at (2,-4.5) {(iii) Distributed};
\end{tikzpicture}
\end{center}
  \caption{\label{fig:setup} Three encoding settings: (i) Centralized -- an
    encoder sends $n$ bits after observing $n$ samples. (ii) Adaptive
    (sequential) -- the $i$th encoder sends the bit $B_i$ depending on its
    private sample $X_i$ and previous bits $B_1,\ldots,B_{i-1}$. (iii)
    Distributed -- each encoder send the bit $B_i$ based on its private
    sample $X_i$ only.}
\end{figure*}

We answer this question in a sylized version of this problem: the estimation
of the mean $\theta$ of a symmetric log-concave distribution under the
constraint that only a single bit can be communicated about each observation
from this distribution.
Different information sharing schemes strongly affect the
performance of estimators for $\theta$; we illustrate
the three main settings we consider in Figure~\ref{fig:setup}.
\begin{enumerate}[(i)]
\item \emph{Centralized} encoding: all $n$ encoders confer and produce a
  single message consists of $n$ bits.
 \item \emph{Adaptive} or \emph{sequential} encoding: The $n$th encoder
   observes the $n$th sample and the $n-1$ previous bits.
 \item \emph{Distributed} encoding: The $n$th message is only a function of
   the $n$th sample.
\end{enumerate}
The distributed setting~(iii) is the most restrictive; as it turns out,
(ii) is slightly more restrictive than the fully centralized setting~(i),
and in our setting, a variant of the adaptive setting~(ii)
in which there is only \emph{one} round of adaptivity---as we make formal
later---is enough to achieve the same efficiency as the fully sequential
setting~(ii).
Each setting has natural applications:
\begin{itemize}
\item {\bf Signal acquisition (i):} A quantity is measured $n$ times at
  different instances. The results are averaged in order to reduce
  measurement noise and the averaged result is then stored or communicated
  using $n$ bits.
\item {\bf Analog-to-digital conversion (ii):} A sigma-delta modulator (SDM)
  converts an analog signal into a sequence of bits by sampling the signal
  at a very high rate and then using one-bit threshold detector combined
  with a feedback loop to update an accumulated error state
  \cite{1092194}. Therefore, the expected error in tracking an analog signal
  using an SDM falls under our setting (ii) when we assume that the signal
  at the input to the modulator is a constant (direct current) corrupted by,
  say, thermal noise \cite{53738}. Since the sampling rates in SDM are
  usually many times more than the bandwidth of its input, analyzing SDM
  under a constant input provides meaningful lower bound even for
  non-constant signals.
\item {\bf Privacy (ii)--(iii):} A business entity is interested in
  estimating the average income of its clients. In order to keep this
  information as confidential as possible, each client independently
  provides an answer to a yes/no question related to its
  income~\cite{DuchiJoWa18}.
\end{itemize}

Let us provide an informal description of our results and setting. For an
estimator $\theta_n$ with finite quadratic risk (mean squared error (MSE)) $R_n =
\E_\theta[(\theta_n - \theta)^2]$, we study the limit
\begin{equation}
  \label{eq:ARE_def}
  \limsup_{n\to \infty} n R_n.
\end{equation}
By comparing this quantity to achievable rates of convergence without
communication constraints, we can evaluate the efficiency
losses---asymptotic relative efficiency---of the estimator to appropriately
optimal (unconstrained) estimators. (We shall be more formal in the sequel.)
By lower bounding the quantity~\eqref{eq:ARE_def}, we also provide limits on
estimation of single-bit-per-measurement constrained signals in more general
settings~\cite{baraniuk2017exponential, jacques2013robust, plan2013one,
  li2017channel, choi2016near}.


In setting (i), the estimator can evaluate any optimal estimator of location
(e.g., the sample mean if the data is Gaussian), then quantize it using
$n$ bits. As the accuracy in describing the empirical mean decreases
exponentially in the number of bits, the quantization error is negligible compared
to the statistical error in mean estimation~\cite{720540, cai2020distributed}. That is,
centralized encoding induces no asymptotic efficiency loss.
The story is different in settings (ii) and (iii). Precisely, we show that
in the adaptive setting~(ii), the optimal efficiency of a one-bit scheme is
(asypmtotically) precisely that of the sample median, and that this
efficiency is achievable. As a concrete example, when $X_i$ are
i.i.d.\ Gaussian, we necessarily lose a factor of $\pi/2 \approx 1.57$ in
the asymptotic risk; the one-bit constraint decreases the effective sample
size by a factor of $\pi/2$ compared to estimating it without the bit
constraint. It turns out that, in the settings we consider,
only a \emph{single round} of adaptivity (see Fig.~\ref{fig:one_round} for
an illustration) is sufficient to achieve optimal convergence rates.
In distinction from setting (ii), in setting (iii) when the messages must be
independent, there is no distributed estimation scheme that achieves the
efficiency of the sample median uniformly over $\theta$.  We establish this
result via Le Cam's local asyptotic normality theory, allowing us
to provide exact characterizations of the asymptotic efficiency
of suitably regular encoding schemes.


Our asymptotic setting is important in that it allows us to elide
difficulties present in finite sample settings. For example, in setting~(i),
developing an optimal quantizer at finite $n$ requires choosing a $2^n$
level scalar quantizer, which is non-trivial~\cite{gray1998quantization}.
In interactive and sequential settings (e.g.~(ii)), the situation is more
challenging, as it is unclear whether any type of compositionality applies,
in that an $n-1$-step optimal estimator may be only vaguely related to the
$n$-step optimal estimator. Thus, to provide our lower bounds, we rely on
stronger information-based inequalities, including the Van Trees
inequality~\cite{Tsybakov09} and Le Cam's local asymptotic normality
theory~\cite{LeCam86,LeCamYa00,VanDerVaart98}.

\subsection*{Related Work}

The many challenges of estimation under communication constraints have given
rise to a large literature investigating different aspects of constrained
estimation. While our setting---in which we observe a single bit per signal
$X_i$---is restrictive, it inspires substantial work.  Perhaps the most
related is that of Wong and Gray~\cite{53738}, who study one-bit
analog-to-digital conversion of a constant input corrupted by Gaussian noise
using a Sigma-Delta Modulator (SDM). They show almost sure convergence, but
provide no rate (and no rates follow from their analysis); in contrast, we
provide an optimal procedure and matching lower bound achieving risk
$\frac{\pi}{2} \sigma^2$ in the limit~\eqref{eq:ARE_def} when $X_i \simiid
\normal(\theta, \sigma^2)$. A growing literature on one-bit measurements in
high-dimensional problems \cite{baraniuk2017exponential, DavenportPlVaWo15,
  PlanVe13} shows how to reconstruct sparse signals, where Baraniuk et
al.~\cite{baraniuk2017exponential} show that in noiseless settings,
exponential decay in MSE is possible; our results make precise the penalty
for noise under one-bit sensing, showing that the error can decay (under
Gaussian noise) at best as $\frac{\pi}{2} \frac{\sigma^2}{n}$.

In fully distributed settings (iii), the challenges are different, and there
is also a substantial literature with one-bit (quantized)
measurements~\cite{904560,4244748, 6882252, chen2010performance, 5184907}.
We complement these results by providing precise lower bounds and optimality
results; previous performance bounds are suboptimal.  Work on the remote
multiterminal source coding problem, or CEO problem~\cite{berger1996ceo,
  viswanathan1997quadratic, oohama1998rate, prabhakaran2004rate}, provides
lower bounds on the MSE in setting~(iii); because of the somewhat distinct
setting, these bounds are looser than ours (which have optimal constants).
In settings more similar to our statistical estimation scenario---such as
estimation of parameters in a multi-dimensional linear model---a line of
work provides lower bounds on statistical
estimation~\cite{zhang2013information, duchi2014optimality, GargMaNg14,
  BravermanGaMaNgWo16, DBLP:journals/corr/abs-1802-08417,
  zhang1988estimation, han2018distributed, xu2017information}. These results
are finite sample and apply more broadly than ours, but as a consequence,
they have unusable constants, while our stylized model allows precise
identification of exact constants.  Work subsequent to the initial draft of
this paper~\cite{Barnes2018} uses an approach similar to ours---bounding
quantized Fisher information---to derive lower bounds on the error in
parametric estimation problems from quantized measurements in non-adaptive
settings.

Testing (and discrete estimation) problems also enjoy a robust literature,
though as a consequence of our results to come, the results for testing,
i.e., when the parameter space $\Theta$ is finite, are quite different from
those for estimation, as it is possible to construct optimal decision
(testing) rules in a completely distributed fashion. In this context, Longo
et al.~\cite{52470} propose procedures for distributed testing based on
optimizing a Bhattacharyya distance.
Tsitsiklis~\cite{tsitsiklis1988decentralized} shows that when the
cardinality of $\Theta$ is at most $M$ and the probability of error
criterion is used, then no more than $M(M-1)/2$ different detection rules
are necessary in order to attain probability of error with optimal exponent.
Moreover, in a distributed setting, feedback is unnecessary for optimal
testing/detection~\cite{5751320}, in strong distinction to the estimation
case we consider.

The remainder of this paper is organized as follows. In
Section~\ref{sec:problem} we describe the problem, notation, and
our basic assumptions. In
Section~\ref{sec:preliminary} we provide two simple bounds on the efficiency
and MSE. Our main results for the adaptive and distributed cases are given
in Sections~\ref{sec:sequential} and \ref{sec:distributed}, respectively. In
Section~\ref{sec:conclusions} we provide concluding remarks.


\section{Problem Formulation and Notation}
\label{sec:problem}

Let $f : \R \to \R_+$ be a symmetric and log-concave probability density,
which necessarily has finite second moment $\sigma^2$, and let $\Theta
\subset \R$ be closed and convex.  For $\theta \in \R$, let $P_\theta$ be
the probability distribution with density $f(x-\theta)$, so that $\theta$
indexes the location family $\{P_\theta\}_{\theta \in \Theta}$.  The
log-concavity and symmetry $f(x)$ imply that $P_\theta$ has a unique mean
and median at $\theta$~\cite{ibragimov1956composition}.
We observe a sample $X_1, \ldots, X_n \simiid P_\theta$, where $\theta$
is unknown, and wish to estimate $\theta$ given only binary
messages $B_1, \ldots, B_n \in \{0, 1\}$ about each $X_i$.
We study this under three distinct computational scenarios, which
we illustrate in Figure~\ref{fig:setup}:
\begin{enumerate}[(i)]
\item \label{item:centralized} Centralized, where $B_i =
  B_i(X_1,\ldots,X_n)$, $i=1,\ldots,n$.
\item \label{item:adaptive} Adaptive, where $B_i =
  B_i(X_i,B_1,\ldots,B_{i-1})$, $i=2,\ldots,n$.
\item \label{item:distributed}
  Distributed, where $B_i = B_i(X_i)$, $i=1,\ldots,n$.
\end{enumerate}
\noindent
We also consider a hybrid of the fully distributed setting (where the bits
$B_i$ are independent) and the adaptive setting (where each bit $B_i$ may
depend on the previous bits) to a \emph{one-step} adaptive setting, where
the quantization scheme may be modified to depend on one fixed function of
the previous information.
\begin{enumerate}[(i')]
\setcounter{enumi}{1}
\item \label{item:one-step-adaptive}
  One-step adaptive, where for some function $g$ and
  a (fixed) $t$, if $i \le t$ then
  $B_i = B_i(X_i)$ while if $i > t$, then
  $B_i = B_i(X_i, g(B_1, \ldots, B_t))$.
\end{enumerate}

We measure the performance of an estimator $\theta_n \defeq
\theta_n(B_1,\ldots,B_n)$ by one of a few notions. In the simplest case,
we assume a prior $\pi$ on $\theta$ (which may be a point
mass) and consider the quadratic risk
\begin{equation}
  \label{eq:error_def}
  R_n = R_n(\pi) \defeq \int \E_\theta\left({\theta}_n - \theta \right)^2
  d\pi(\theta),
\end{equation}
where the expectation is taken with respect to the distribution of
$X_1,\ldots,X_n \simiid P_\theta$.  The main problems we consider in this
paper are the minimal value of the risk~\eqref{eq:error_def} as a function
of the sample size $n$ and the density $f$, under different choices of the
encoding functions in
cases~\eqref{item:centralized}--\eqref{item:distributed}.
The quadratic risk~\eqref{eq:error_def} may be infinite in some cases;
we defer discussion of this case to later sections, as it is technically
demanding and detracts from the presentation here.

Now, let $\sigma_f^2 \defeq \E[\frac{f'(X)^2}{f(X)^2}]$ be the Fisher
information for the location in the family $\{P_\theta\}$, which is finite
when $f$ is log-concave and symmetric. We give particular attention to the
asymptotic relative efficiency (ARE) of estimators with respect to
asymptotically normal efficient estimators achieving the information
bound~\cite{VanDerVaart98}. In this case,
if $\{m(n), n \in \N\}$ is a sequence such that
\begin{equation*}
  \sqrt{m(n)} (\theta_n - \theta) \cd \normal(0, \sigma_f^2),
\end{equation*}
then the ARE of the estimator is~\cite[Def.~6.6.6]{LehmannCa98}
\begin{equation}
  \label{eqn:are-def}
  \ARE({\theta}_n) \defeq
  \liminf_{n\rightarrow \infty} \frac{m(n)}{n}.
\end{equation}
In the special case where there exists $V \in \mathbb R$ such that
\begin{equation*}
  m(n) R_n =
  m(n) \mathbb E_\theta\left({\theta}_n - \theta \right)^2 = V + o(1),
\end{equation*}
the ARE of ${\theta}_n$ is $\sigma_f^2/V$, so that $\theta_n$ requires a
sample $V / \sigma_f^2$-times larger than that of an efficient estimator for
comparable accuracy to the (information) efficient estimator.

\subsection*{Notation and basic assumptions}

To describe our results and make them formal, we require some additional
notation and one main assumption, which restricts the class of distributions
we consider.  We use the typical notation that
$F(x) = \int_{-\infty}^x f(t) dt$ is the cumulative distribution function
of the $X_i$, and we let
\begin{equation*}
  h(x) \triangleq \frac{f(x)}{1-F(x)} = \frac{f(x)}{F(-x)}
\end{equation*}
be the \emph{hazard} function (or the \emph{failure rate} or \emph{force of
  mortality}), which is monotone increasing as $f$ is
log-concave~\cite{bagnoli2005log}. Given the centrality of the median
to our efficiency bounds, it is unsurprising that the quantity
\begin{equation}
  \label{eq:eta_def}
  \eta(x) \triangleq \frac{f^2(x)}{F(x)(1-F(x))}
  \stackrel{(\star)}{=} \frac{f(x)f(-x)}{F(x)F(-x)}
\end{equation}
appears throughout our development (equality~$(\star)$ is immediate by
the symmetry of $f$). For $p \in (0, 1)$ and $x = F^{-1}(p)$,
\begin{equation}
  \label{eqn:variance-quantiles}
  \frac{1}{\eta(x)} =
  \frac{1}{\eta(F^{-1}(p))}
  = \frac{p (1 - p)}{f(F^{-1}(p))^2}
\end{equation}
is of course the familiar asymptotic variance of the $p$th quantile of the sample $X_1,\ldots,X_n$ (cf.~\cite{VanDerVaart98}, Ch.~21).

For $f$ the normal density, classical results~\cite{Samford1953,
  hammersley1950estimating} show that $\eta(x)$ is a strictly
decreasing function of $|x|$, as we illustrate in Fig.~\ref{fig:eta}.
We consider log-concave symmetric distributions sharing this
property.  Specifically, we require the following.
\begin{assumption}
  \label{assump:failure_rate}
  The density $f$ is log-concave and symmetric.  Additionally, the origin $x
  = 0$ uniquely maximizes $\eta(x)$, and $\eta(x)$ is non-increasing in
  $|x|$.
\end{assumption}
Under this assumption,
\begin{equation*}
  4 f^2(x) \leq \eta(x) \leq \eta(0),
\end{equation*} 
where $\eta(0) = 4 f^2(0)$ is the asymptotic variance of the sample median
(Eq.~\eqref{eqn:variance-quantiles} at $p = 1/2$).  Combined with
log-concavity of $f(x)$, Assumption~\ref{assump:failure_rate} implies that
$\eta(x)$ vanishes as $|x|\rightarrow \infty$.  Several distributions
satisfy Assumption~\ref{assump:failure_rate}, including the generalized
normal distributions with a shape parameter between $1$ and $2$ (including
the normal and Laplace distributions). Symmetric log-concave distributions
failing Assumption~\ref{assump:failure_rate} include the uniform
distribution and the generalized normal distribution with shape parameter
greater than $2$. Some restriction on the class of distributions is
necessary to develop our results; indeed, in
Appendix~\ref{sec:uniform-weirdos} we provide a brief discussion on the
uniform distribution, where a one-step adaptive estimator with single bit
observations can achieve convergence rates faster than the familiar
$\sqrt{n}$ paramateric rate.

\begin{figure}
\begin{center}
\begin{tikzpicture}[scale = 0.6]
\begin{axis}[
width=8cm, height=6cm,
xmin = -3, xmax=3, 
restrict y to domain = 0:100,
ymin = 0,
ymax = 0.9,
samples=10, 
xlabel= $x$,
xtick={-2,-1,0,1,2},
xticklabels={-2,-1,0,1,2},
ytick={0,0.3989423,0.6366198},
yticklabels={0,$\frac{1}{\sqrt{2 \pi}}$,$2/\pi$},
line width=1.0pt,
mark size=1.5pt,
ymajorgrids,
xmajorgrids,
legend style= {at={(1,1)},anchor=north east,draw=black,fill=white,align=left}
]
\addplot[color = blue, solid, smooth] plot table [x = x, y = y, col sep=comma] {./Figs/eta.csv};
\addlegendentry{$\eta(x)$};
\addplot[domain = -5:5, samples = 50, color = red, solid, smooth]  {exp(-x^2/2) / sqrt(2*3.14159)};
\addlegendentry{$\phi(x)$};
\addplot[domain = -5:5, samples = 50, color = black, solid, dashed]  {4*exp(-x^2) / (2*3.14159)};
\addlegendentry{$4\phi^2(x)$};

\end{axis}
\end{tikzpicture}
\begin{tikzpicture}[scale = 0.6]
\begin{axis}[
width=8cm, height=6cm,
xmin = -4, xmax=4, 
restrict y to domain = -10:0,
ymin = -10,
samples=10, 
xlabel= $x$,
xtick={-3,-2,-1,0,1,2,3},
xticklabels={-3,-2,-1,0,1,2,3},
ytick={0,-0.45158,-0.919},
yticklabels={0,,},
line width=1.0pt,
mark size=1.5pt,
ymajorgrids,
xmajorgrids,
legend style= {at={(1,1)},anchor=north east,draw=black,fill=white,align=left}
]
\addplot[color = blue, solid, smooth] plot table [x = x, y = logy, col sep=comma] {./Figs/eta.csv};
\addlegendentry{$\log \eta(x)$};
\addplot[domain = -5:5, samples = 30, color = red, solid, smooth]  {-(x)^2/2 -0.9189};
\addlegendentry{$\log \phi(x)$};
\end{axis}
\end{tikzpicture}
\caption{
The function $\eta(x) = f^2(x) / F(x)F(-x)$ for $f(x) = \phi(x)$ the standard normal density.
\label{fig:eta}
}
\end{center}
\end{figure}

\section{Consistent Estimation and Off-the-shelf Bounds \label{sec:preliminary}}

We begin our technical treatment by deriving a few
bounds on the efficiency of estimators in
setting~\eqref{item:distributed}. These bounds establish the
following facts:
\begin{enumerate}[1.]
\item A consistent estimator with an asymptotically normal distribution
  always exists in setting~\eqref{item:distributed}, and hence in the
  adaptive settings~\eqref{item:adaptive} and
  (\ref{item:one-step-adaptive}').
\item For the normal distribution, the asymptotic relative
  efficiency~\eqref{eqn:are-def} in the distributed
  setting~\eqref{item:distributed} is at most $3/4$. No estimator can be as
  efficient as the sample mean.
\end{enumerate}

\subsection{Consistent Estimation}
The simplest estimator is simply to invert a quantile. Indeed,
fix $\theta_0 \in \mathbb R$ and define the $i$th message by 
\[
B_i = \indic{X_i<\theta_0}, 
\]
where $\indic{A}$ is the indicator of the event $A$. We have
\[
\bar{B}_n \defeq\frac{1}{n} \sum_{i=1}^n B_i \overset{a.s.}{\rightarrow} F(\theta_0 - \theta),  
\]
so that 
\begin{equation}
\label{eq:estimator_naive}
{\theta}_n = \theta_0 - F^{-1}\left( \bar{B}_n \right)
\end{equation}
is a consistent estimator for $\theta$ in the distributed setting of
Figure~\ref{fig:setup}-(iii), where we note that $F$ is invertible over the
support of $f$. As the variance of $\bar{B}_n$ is
$F(\theta_0-\theta)\left(1-F(\theta_0-\theta)\right)$, a delta method
calculation~\cite[Ch.~23]{VanDerVaart98} implies that ${\theta}_n$ is
asymptotically normal with variance
\begin{equation*}
  \frac{F(\theta_0-\theta)\left(1-F(\theta_0-\theta)\right)}{f^2(\theta_0-\theta)} = \frac{1}{
\eta(\theta_0-\theta)}.
\end{equation*}
In the Gaussian case where the $X_i \simiid \normal(\theta, \sigma^2)$, the
ARE of ${\theta}_n$ is $\eta(\theta_0 - \theta)\sigma^2$.

Assumption~\ref{assump:failure_rate} implies that the optimal asymptotic
variance for an estimator of the form~\eqref{eq:estimator_naive} is $1 /
\eta(0)$, the asymptotic of the sample median. Unfortunately, as $\theta$ is
(by definition) \emph{a priori} unknown and $\eta(x)$ monotonically
decreases in $|x|$, this naive estimator $\theta_n$ may be very inefficient
when $\theta$ is far from the initial guess $\theta_0$. As an example, when
$f$ is a the normal density, the ARE of ${\theta}_n$ is less than $0.15$
when $|\theta_0 - \theta| \ge 2\sigma$, and more broadly, $\ARE(\theta_n)$
asymptotes to $|\theta_0| \exp(-\theta_0^2 / 2) / \sqrt{2\pi}$ as $|\theta_0
- \theta|$ gets large.  Yet that $\theta_0 = \theta$ minimizes this
asymptotic variance, and $\eta$ is continuous, is suggestive: if we can use
a suitably good initial estimate $\theta_n\init$ for $\theta$, it is
possible that a one-step adaptive estimator
(recall~(\ref{item:one-step-adaptive}')) may be asymptotically strong, as we
see in Section~\ref{sec:sequential}.

\subsection{Multiterminal Source Coding}
\label{sec:ceo}

A related problem is the CEO problem, which considers the estimation of a
sequence $\theta_1,\theta_2\ldots$, where a noisy version of each $\theta_j$
is available at $n$ terminals. At each terminal $i$, an encoder observes the
$k$ noisy samples
\[
X_{i,j} = \theta_j + Z_{i,j},\qquad j=1,\ldots,k, \qquad i = 1,\ldots,n,
\]
and transmits $r_i k$ bits to a central estimator~\cite{berger1996ceo}. The
central estimator produces estimates ${\hat{\theta}}_1,\ldots,{\hat{\theta}}_k$ with
the goal of minimizing the quadratic risk:
\[
R_{\CEO} = \frac{1}{k} \sum_{j=1}^k \mathbb E \left[\left(\hat{\theta}_j - {\theta_j} \right)^2 \right]. 
\]
Note that any distributed encoding scheme using one-bit per sample can be replicated $k$ times and thus leads to a legitimate encoding and estimation scheme for the CEO problem with $r_1=\ldots=r_n = 1$. It follows that, assuming that $\theta$ is drawn once from the prior $\pi$, our mean estimation problem from one-bit samples under distributed encoding corresponds to the CEO setting with $k=1$ realization of $\theta$ observed under noise at $n$ different locations, and communicated at each location using an encoder sending a single bit. 
Consequently, a lower bound on the MSE in estimating $\theta$ in the
distributed encoding setting is given by the minimal MSE in the CEO setting
as $k \to \infty$. Note that the difference between the CEO setting and ours
lays in the privilege of each of the encoders to describe $k$ realizations
of $\theta$ using $k$ bits with MSE averaged over these realizations, rather
than a single realization using a single bit in ours.

When the prior on $\theta$ and the noise corrupting it at each location are
Gaussian, Prabhakaran et al.~\cite{prabhakaran2004rate} characterize the
optimal encoding and its asymptotic risk as $k \to \infty$.  Chen et
al.~\cite{chen2004upper} also provide an expression for the quadratic risk
in the CEO setting under Gaussian priors. Adapting to our setting, this
expression provides the following proposition:
\begin{prop} \label{prop:CEO}
  Assume that $\Theta = \mathbb R$ and $\pi(\theta) =
  \Ncal(0,\sigma_\theta^2)$ where $\sigma_\theta^2 \in \mathbb R$ is
  arbitrary. Then any estimator ${\theta}_n$ of $\theta$ in the distributed
  setting satisfies
  \begin{equation} \label{eq:ceo_bound}
    n \cdot \ex{\left( \theta - \theta_n \right)^2} \geq \frac{4}{3} \sigma^2 + O(n^{-1}),
  \end{equation}
  where the expectation is with respect to $\theta$ and $X_1,\ldots,X_n$.
\end{prop}
\noindent
See Appendix~\ref{app:proof:CEO} for a proof.

As we shall see, this bound is loose: the difference between the MSE lower
bound~\eqref{eq:ceo_bound} and the actual MSE in the distributed setting
(case~\eqref{item:distributed}) occurs because in the CEO setting, each
encoder may encode an arbitrary number of $k$ independent realizations of
$\theta$ using $k$ bits; in our situation, $k = 1$. That blocking allows
more efficient encoding and exploiting the high-dimensional geometry of the
product probability space in the CEO problem is perhaps unsurprising, and
our goal in the sequel will be to characterize the performance degradation
one bit encoding engenders.

\section{Adaptive Estimation \label{sec:sequential}}

The first main result of this paper (Theorem~\ref{thm:adpative_lower_bound})
gives that the asymptotic variance of any adaptive estimator must be at
least $\eta(0)\sigma^2$, which is precisely the efficiency of the median of
the sample $X_1,\ldots,X_n$. Conveniently, the stochastic (sub)gradient
estimator for the median---which minimizes $\E[|X - \theta|]$---is a
sequence of signs (single bits), so that we can exhibit an asymptotically
optimal adaptive estimation scheme.


We begin with our first theorem, whose proof we provide in
Appendix~\ref{proof:thm:adpative_lower_bound}.
\begin{thm}[Fundamental limits]\label{thm:adpative_lower_bound}
  Let Assumption~\ref{assump:failure_rate} hold.
  Let ${\theta}_n$ be any estimator of $\theta$ in the adaptive setting of
  Figure~\ref{fig:setup}(ii). Assume that the prior
  density $\pi(\cdot)$ on $\theta$ converges to zero
  at the endpoints of the interval $\Theta$ and
  define the prior Fisher information
  $I_0 \defeq \E_\pi[(\pi'(\theta) / \pi(\theta))^2]$.
  Then
  \begin{equation*}
    \E\left[ (\theta-{\theta}_n)^2 \right] \geq   \frac{1}{ 4f^2(0) n + I_0}.
  \end{equation*}
\end{thm}

We now turn to asymptotically optimal estimators, first
showing how a simple stochastic gradient scheme is asymptotically
optimal (in the fully adaptive setting), after which we show that
a one-round adaptive scheme can also achieve this optimal efficiency.

\subsection{Asymptotically optimal estimator}

The starting point for our first estimator is to note that the median of a
distribution minimizes $\E[|X - \theta|]$ over $\theta \in \R$, and
moreover, we have the familiar result (cf.~\cite{VanDerVaart98}, Ch.~21)
that given a sample $X_1,\ldots,X_n \simiid P$, if $\theta = \mbox{med}(P)$ and
$P$ has continuous density $f(\cdot - \theta)$ near $\theta$, then
\begin{equation*}
  \sqrt{n}(\mbox{med}(X_1^n) - \theta)
  \cd \normal\left(0, \frac{1}{4 f(0)^2}\right),
\end{equation*}
which is precisely the variance lower bound in
Theorem~\ref{thm:adpative_lower_bound}.  Thus, it is natural to consider a
stochastic gradient procedure for minimizing $\E[|X - \theta|]$. To that end,
let $\left\{ \gamma_n \right\}_{n\in \mathbb N}$ be a strictly positive
sequence of stepsizes,
and define the sequence
\begin{equation}
  \label{eq:sgd_alg}
  \theta_n = \theta_{n-1} + \gamma_n B_n, \quad n = 1,2,\ldots,
\end{equation}
where 
\begin{equation*}
  B_n = \sgn (X_n - \theta_{n-1}).
\end{equation*}
We make one of two assumptions on the stepsizes $\gamma_n$, which
are relatively standard: we always have $\gamma_n$ non-increasing, and
\begin{subequations}
  For some $0 < \lambda \le 1$,
  \begin{align}
    \label{eqn:lazy-gamma}
    \frac{\gamma_n - \gamma_{n+1}}{\gamma_n^2}
    \to 0, & ~~~
    \sum_n \frac{\gamma_n^\frac{1 + \lambda}{2}}{\sqrt{n}} < \infty
    ~~ \mbox{or} \\
    \gamma_n = o(n^{-2/3}),
    & ~~~
    \sum_n \gamma_n = \infty.
    \label{eqn:stringent-gamma}
  \end{align}
\end{subequations}

Then we can adapt the results of Polyak and Juditsky~\cite{polyak1992acceleration}
on the asymptotic normality of averaged stochastic gradient estimators
to establish the following theorem.
\begin{thm}
  \label{thm:sgd}
  Define the average $\bar{\theta}_n \defeq \frac{1}{n}
  \sum_{i = 1}^n \theta_i$. Assume
  that in a neighborhood
  of $\theta = \mbox{med}(P)$,
  the distribution $P$ has a Lipschitz continuous density $f$.
  Then
  \begin{enumerate}[(i)]
  \item \label{item:normal-sgd}
    Assume that $\left\{ \gamma_n \right\}_{n\in \mathbb N}$ satisfies
    condition~\eqref{eqn:lazy-gamma}.
    Then
    \begin{equation*}
      \sqrt{n}\left( \bar{\theta}_n - \theta\right)
      \cd \normal\left(0,\frac{1}{4 f(0)^2}\right).
    \end{equation*}
  \item \label{item:sgd-regular}
    Let $\{P_\theta\}_{\theta \in \R}$ be the family of distributions
    with density $f(\cdot - \theta)$, where $f$ has median 0.
    Let $h_n \to h \in \R$, and define the distributions
    $P_n = P_{\theta + h_n/\sqrt{n}}^n$. Then
    \begin{equation*}
      \sqrt{n}\left(\bar{\theta}_n - \theta - h_n / \sqrt{n}\right)
      \mathop{\cd}_{P_n}
      \normal\left(0, \frac{1}{4 f(0)^2}\right),
    \end{equation*}
    and for any bounded, symmetric, and quasi-convex function $L$,
    \begin{align} 
      & \sup_{c < \infty} \limsup_{n \to \infty}
      \sup_{\tau\,:\,|\theta-\tau| \leq \frac{c}{\sqrt{n} }}
      \E_\tau \left[ L\left( \sqrt{n}(\bar{\theta}_{n} - \tau) \right) \right] \nonumber 
      \\
      & \qquad \qquad \qquad \qquad = \mathbb E \left[L(Z/ 2 f(0)) \right],
        \label{eq:attaining_LAM}
    \end{align}
    where $Z \sim \normal(0,1)$. 
  \item \label{item:sgd-ms-convergence} Assume the stepsizes $\gamma_n$
    satisfy both conditions~\eqref{eqn:lazy-gamma}
    and~\eqref{eqn:stringent-gamma}. Let
    $\pi$ be a distribution on $\R$ \newtext{with a finite second moment}. Then
    \begin{align}
      \int \E\Big[( \bar{\theta}_n - \theta )^2\Big] \pi(d\theta)
      = \frac{1}{4 n f(0)^2} + o(n^{-1}).
      \label{eq:adaptive_3}
    \end{align}
  \end{enumerate}
\end{thm}

\noindent
We provide the proofs of items (i)-(iii) in Appendices~\ref{sec:proof-normal-sgd},
\ref{sec:proof-sgd-regular}, \ref{sec:proof-sgd-ms-convergence},
respectively.

As an immediate corollary to Theorem~\ref{thm:sgd}, we obtain the following
asymptotic optimality results of the averaged stochastic gradient
sequence. Specifically, the average of the stochastic gradient
iterates~\eqref{eq:sgd_alg}
is locally asymptotically minimax, and they achieve the lower
bound of Theorem~\ref{thm:adpative_lower_bound}.

\begin{corollary}
  Let the conditions of Theorem~\ref{thm:adpative_lower_bound} hold
  and $\theta_n$ be defined by the iteration~\eqref{eq:sgd_alg}.
  Let $\{P_\theta\}_{\theta \in \R}$ be the family of distributions
  with densities $f(\cdot - \theta)$.
  \begin{enumerate}[(i)]
  \item Define the shorthand $P_n = P_{\theta + h_n/\sqrt{n}}^n$.
    If the stepsizes satisfy condition~\eqref{eqn:lazy-gamma}, then
    \begin{equation*}
      \sqrt{n}(\bar{\theta}_n - \theta - h_n / \sqrt{n})
      \mathop{\cd}_{P_n} \normal\left(0, \frac{1}{\eta(0)}\right).
    \end{equation*}
  \item If in addition the stepsizes satisfy
    condition~\eqref{eqn:stringent-gamma}, then they
    achieve the lower bound of Theorem~\ref{thm:adpative_lower_bound} for any
    prior $\pi$ on $\R$.
  \end{enumerate}
\end{corollary}

\begin{figure}
\begin{center}
\begin{tikzpicture}[node distance=2cm,auto]
  \node at (0,0) (source) {$X_1$} ;
  \node[int1, right of = source, node distance = 1.2cm] (enc1) {$\enc$};  
\draw[->,line width = 2pt] (source) -- (enc1); 


\node[below of = source, node distance = 1.7cm] (source3) {$X_{n_1}$};
\node[int1, right of = source3, node distance = 1.2cm] (enc3) {$\enc$};  

\draw[->,line width = 2pt] (source3) -- (enc3); 

\node[below of = source, node distance = 0.5cm] {$\vdots$};

\node[int1, right of = enc3, node distance = 2.1cm ] (est) {$\est_1$};

\draw[->,line width = 2pt] (enc1) -| node[above, xshift = -1cm] (mes1) {$B_1$} (est);   


\draw[->,line width = 2pt] (enc3) -- node[above, xshift = 0cm]  {$B_{n_1}$} (est);   

\node[below right = 0.75 and 1.5 of source3] (sourceB) {$X_{n_1 +1}$} ;
\node[int1, right of = sourceB, node distance = 1.7cm] (enc1B) {$\enc$};  
\draw[->,line width = 2pt] (sourceB) -- (enc1B); 

\node[below of = sourceB, node distance = 1.7cm] (source3B) {$X_n$};
\node[int1, right of = source3B, node distance = 1.7cm] (enc3B) {$\enc$};  
\draw[->,line width = 2pt] (source3B) -- (enc3B); 
\node[below of = sourceB, node distance = 0.4cm] {$\vdots$};

\node[int1, right of = enc3B, node distance = 2.1cm ] (estB) {$\est_2$};

\draw[->,line width = 2pt] (enc1B) -| node[above, xshift = -0.5cm] {$B_{n_1+1}$} (estB);
\draw[->,line width = 2pt] (enc3B) -- node[above] {$B_n$} (estB);

\draw[->,line width = 1pt] (est.east) node[above, xshift  =0.5cm] {${\theta}_{n_1}$} -| (enc1B.north);


\draw[->,line width = 0.5pt] (est.east) -| +(1.3,-0.5) -- +(1.3,-2.5) -| (enc3B.north);


\draw[->,line width = 0.5pt] (estB) -- +(0.8,0) node[right] {${\theta}_n$};
\node[below of = enc1B, node distance = 0.5cm] {$\vdots$};

\end{tikzpicture}
\end{center}
\caption{Distributed encoding with one round of threshold adaptation. The estimation obtained from the first $n_1$ bits in a distributed manner is utilized in obtaining another $n-n_1$ bits in a distributed manner. 
\label{fig:one_round}
}
\end{figure}

\subsection{Maximal Efficiency using One Round of Threshold Adaptation}

In the encoding and estimating procedure \eqref{eq:sgd_alg}, each one-bit
message $B_n$ depends on its private sample as well as the current gradient
descent estimate $\theta_{n-1}$. In this sense, each encoder in this
algorithm interacts with previous one by using the current estimate.  This
amount of adaptivity is unnecessary: as we now consider, a similar encoding
yields an asymptotically normal estimator attaining the lower variance bound
$1/\eta(0)$, provided we allow \emph{one} adaptive update to the threshold
value $\theta_0$ based on previously observed bits.
In this procedure we separate the sample into the disjoint sets
$X_1,\ldots,X_{n_1}$ and $X_{n_1+1},\ldots,X_n$ for some $n_1 < n$.  We
first use the estimator \eqref{eq:estimator_naive} to obtain an estimate
${\theta}_{n_1}$ based on $B_1,\ldots,B_{n_1}$, and then use
${\theta}_{n_1}$ as the new threshold value to obtain messages $B_{n_1+1},
\ldots, B_n$. Figure~\ref{fig:one_round} illustrates a diagram of this
procedure.

More formally, we consider the following estimation scheme. Given
$n_1 \in \{1, \ldots, n\}$,
set the individual bits
\begin{equation*}
  B_i =
  \begin{cases}
    \indic{X_i \le \theta_0} & i = 1,\ldots,n_1, \\
    \indic{X_i \le T_n}& i={n_1+1,\ldots,n},
  \end{cases}
\end{equation*}
where
\begin{align*}
  T_n & \defeq \theta_0 - F^{-1}\left(
  \frac{1}{n_1} \sum_{i=1}^{n_1} B_i 
  \right)\\
  \theta_n & \defeq
  T_n - F^{-1} \left(\frac{1}{n - n_1}
  \sum_{i = n_1}^n B_i \right).
\end{align*} 
The intuition here is that the estimator $\theta_n$ is a one-step
correction (cf.~\cite[Thm.~6.4.3]{LehmannCa98}) of the initial estimator
$T_n$, which approximately estimates
$\theta_0 - F^{-1}\left(F(\theta_0 - \theta)\right) = \theta$. We then have the
following convergence result.
\begin{thm}
  Assume that $X_i = Z_i + \theta$, where $Z_i$ are i.i.d.\
  with density $f$ and CDF $F$ and $\mbox{med}(Z_i) = 0$. Assume that
  $f$ is continuous at 0, and that as $n \to \infty$,
  $n_1(n) \rightarrow \infty$ and $n_1 / n \to
  0$. Then
  \begin{align*}
    \sqrt{n} \left( {\theta}_n - \theta  \right)
    \cd  \normal\left( 0, \frac{1}{4 f(0)^2}\right).
  \end{align*}
\end{thm}
\noindent
That is, under Assumption~\ref{assump:failure_rate}, the method is
asymptotically optimal.
\begin{proof}
  We abuse notation and instead of assuming we receive $n$ observations, assume
  we receive the $n + n_1$ observations $X_{-n_1}, \ldots, X_{-1}$ and
  $X_1, \ldots, X_n$, defining $T_n = \theta_0 - F^{-1}(\frac{1}{n_1}
  \sum_{i = -n_1}^{-1} B_i)$ and
  $B_i = \indic{X_i \le T_n}$ for $i \ge 1$.
  Letting $X_i = Z_i + \theta$ for $Z_i$ i.i.d.\ with fixed density
  $f = F'$, we have
  $\ex{B_i} \cas F(\theta_0 - \theta)$, so that
  $\frac{1}{n_1} \sum_{i = -n_1}^{-1} B_i \cas F(\theta - \theta_0)$ and
  by the continuous
  mapping theorem we have $T_n \cas \theta$ as $n_1 \to \infty$.

  Now let
  $E_n = \E[B_i \mid T_n] = P(X_i \le T_n)$,
  so that $\var(B_i \mid T_n) = E_n(1 - E_n)$. Define also the random
  variable
  \begin{equation*}
    Y_n \defeq \sqrt{n}
    \frac{1}{\sqrt{E_n (1 - E_n)}}
    \bigg[\frac{1}{n}\sum_{i = 1}^n B_i - E_n\bigg],
  \end{equation*}
  and let $F_n(\cdot \mid T_n)$ be its cumulative distribution function.
  Then because
  \begin{equation*}
    \E\left[|B_i - E_n|^3 \mid T_n\right] \le E_n(1 - E_n),
    \end{equation*}
      we have 
\begin{equation*}
    \E\left[\frac{|B_i - E_n|^3}{(E_n(1 - E_n))^{3/2}} \mid T_n\right]
    \le \frac{1}{\sqrt{E_n(1 - E_n)}}.
  \end{equation*}
 The Berry-Esseen theorem implies that there exists a constant
  $C \le 1$ such that
  \begin{equation*}
    \sup_t \left|F_n(t \mid T_n) - \Phi(t) \right|
    \le \frac{C}{\sqrt{E_n (1 - E_n)} \sqrt{n}} \wedge 2,
  \end{equation*}
  where $\Phi$ is the standard Gaussian CDF.
  As $E_n (1 - E_n) \cas \frac{1}{4}$ by definition of the median,
  we have
  that (with probability 1)
  \begin{equation*}
    \sup_t \left|F_n(t \mid T_n) - \Phi(t)\right| \le \frac{C}{\sqrt{n}}
    ~~ \mbox{eventually}.
  \end{equation*}
  By dominated convergence and Jensen's inequality we thus obtain
  \begin{equation*}
    \sup_t \left|\mathbb{P}(Y_n \le t) - \Phi(t)\right|
    \le \E\left[
      \sup_t \left|F_n(t \mid T_n) - \Phi(t)\right| \right]
    \to 0,
  \end{equation*}
  which gives that $Y_n \cd \normal(0, 1)$. Now, Slutsky's lemmas imply
  \begin{align}
    &\sqrt{n}
    \cdot \frac{2}{n} \sum_{i = 1}^n (B_i - E_n)
    \\
    & \qquad  =
 \frac{1 + o_P(1)}{\sqrt{n E_n(1 - E_n)}}
    \sum_{i = 1}^n \left(B_i - E_n\right)
    \cd \normal(0, 1).
    \label{eqn:apply-slutsky}
  \end{align}
  where $o_P(1)$ denotes sequence of random variables converging to zero in probability as $n$ goes to infinity. With $\bar{B}_n \defeq \frac{1}{n} \sum_{i = 1}^n B_i$ and using that $E_n = \E[B_i \mid T_n] = F(T_n - \theta)$, we may use the delta method to write
  \begin{align*}
    & \sqrt{n}(\theta_n - \theta)
     = \sqrt{n}\left(T_n - F^{-1}\left( \bar{B}_n \right)
    - \theta \right) \\
    &\quad  = \sqrt{n} \left[T_n - F^{-1}\left(F(T_n - \theta)
    + \bar{B}_n - F(T_n - \theta)\right) - \theta \right] \\
    & \quad = \sqrt{n} \left[ T_n - (T_n - \theta ) \right. \\
    &  \left. +  (F^{-1})' \left(T_n - \theta + o_P(1) \right)  
    \cdot \left(\bar{B}_n - E_n\right) - \theta \right] \\
    & \quad = \sqrt{n} (F^{-1})'(0) ( \bar{B}_n - E_n)
    + o_P(1) \\
    & \quad \cd \normal\left(0, \frac{1}{4 f(0)^2}\right),
  \end{align*}
  where we have used the limiting distribution~\eqref{eqn:apply-slutsky}. 
\end{proof}

Figure \ref{fig:adaptive_error} illustrates the empirical risks 
of the estimator \eqref{eq:sgd_alg} and an estimator obtained using one round of threshold adaptation under a series of Monte Carlo simulations when $f(x)$ is the standard normal desnity.

\begin{figure}
\begin{center}
\begin{tikzpicture}[scale = 0.6]
\begin{axis}[
width=10cm, height=6cm,
xmin = 200, xmax=800, 
restrict y to domain = 0:3,
ymin = 0,
ymax = 3.4,
samples=10, 
xlabel= $n$,
ylabel = {$n\cdot\mathbb E \left[\left(\theta - {\theta}_n \right)^2 \right]$},
ytick={0,1,1.57},
yticklabels={0,1,$\pi/2$},
line width=1.0pt,
mark size=1.5pt,
ymajorgrids,
xmajorgrids,
legend style= {at={(1,1)},anchor=north east,draw=black,fill=white,align=left}
]

\addplot[color = blue, solid, smooth] plot table [x = itr, y = SGD, col sep=comma] {./SimRes/sim_res_nMonte5000.csv};
\addlegendentry{asymptotically optimal};


\addplot[color = red, solid, smooth] plot table [x = itr, y = split, col sep=comma] {./SimRes/sim_res_nMonte5000.csv};
\addlegendentry{one adaptation};

\end{axis}
\end{tikzpicture}
\caption{Normalized empirical risk versus number of samples $n$ for $10,000$ Monte Carlo trials with $f(x)$ the standard normal density. In each trial, $\theta$ is chosen uniformly over the interval $(-1.64,1.64)$. The one round threshold adaptation strategy uses $n_1 = \lfloor \sqrt{n} \rfloor$ samples before adapting the threshold.
\label{fig:adaptive_error}  }
\end{center}
\end{figure}

\section{Distributed Estimation \label{sec:distributed}}

We now consider the distributed encoding setting in
Figure~\ref{fig:setup}-(iii) where each one-bit message $B_i$ is a
function only of its private sample $X_i$. In this case, the $i$th encoder is
of the form $B_i = \indic{X_i \in A_i}$, where the detection region
$A_i$ is a Borel set independent of $X_1, X_2, \ldots$.

\subsection{Optimal Efficiency}
We begin by making a few restrictions on the collections of the sets $A_i$,
which we believe not unreasonable, but which allow us to develop fundamental
limits for estimation. We require a bit of notation to define the
assumptions. As we work with a location family based on a density $f$
with associated probability distribution $P$ on variables $Z$,
we define
\begin{equation*}
  P_\theta(A) \defeq P(Z - \theta \in A)
\end{equation*}
for $Z$ with density $f$. Whenever $A$ is a collection of disjoint intervals $A = \cup_i [t_i^-, t_i^+]$, we may define
\begin{equation*}
  \dPtheta(A) \defeq \frac{\partial}{\partial \theta} P_\theta(A)
  = \sum_i \left(f(t_i^- - \theta) - f(t_i^+ - \theta)\right),
\end{equation*}
and similarly we define the score
function $\score_\theta(A) \defeq \dPtheta(A) / P_\theta(A)$.
For $B = \indic{X \in A}$, we abuse notation and also write
$\score_\theta(B) = \score_\theta(A)$ and similarly for $\dPtheta$.
With this, we may define the variance of the
scores $\score_\theta(B_i)$ under $P_\theta$ via
\begin{equation}
  \label{eq:precision_general}
  L_n(A_1,\ldots,A_n;\theta) \defeq
  \frac{1}{n} \sum_{i=1}^n \frac{ \dPtheta(A_i)^2}{
    P_\theta(A_i) (1 - P_\theta(A_i))}.
\end{equation}
We then make the following assumption.
\begin{assumption}
  \label{assumption:detection-regions}
  The density and detection regions satisfy
  \begin{enumerate}[(i)]
  \item \label{item:lipschitz-density}
    The density function $f$ of $X_n - \theta$ is Lipschitz continuous.
  \item \label{item:finite-intervals}
    Each set $A_i$ is the finite union of $k_i$ disjoint intervals
    (which may include $\pm \infty$), where
    \begin{equation*}
      \frac{1}{n} \cdot \max_{i \le n} \frac{k_i^3}{P_\theta(A_i)^4
        (1 - P_\theta(A_i))^4} \to 0.
    \end{equation*}
  \item \label{item:limit-variance}
    The limit
    \begin{equation}
      \label{eqn:LAN-limit}
      \kappa(\theta) \defeq \lim_{n\to \infty} L_n(A_1,\ldots,A_n; \theta)
    \end{equation}
    exists and is finite.
  \end{enumerate}
\end{assumption}
Roughly speaking, \eqref{item:finite-intervals} above holds whenever the intervals consisting each $A_i$ are appropriately seperated and their number is relatively small. For example, it applies when each set $A_i$ is a half-bounded interval $(t_i,\infty)$ with $\min\{P_\theta((t_i,\infty)), P_\theta((-\infty,t_i])\} = \omega(1/n)$ as we dicscuss in more detail below. More generally, let $\Delta_i$ the minimal distance between any two interval endpoints in $A_i$. Then, 
if $A_i = \cup_{j=1}^{k_i}[t_{i,j}^-, t_{i,j}^+]$, we have that 
$P_\theta(A_i) \ge \Delta_i \sum_{j=1}^{k_i} F(t_{i,j}^-)$ and $1-P_\theta(A_i) \ge \Delta_i \sum_{j=1}^{k_i} F(t_{i,j}^+)$. Therefore, \ref{assumption:detection-regions}\eqref{item:finite-intervals} holds whenever $ \max_{i \le n} k_i^3 \Delta_i^{-4} = o(n)$ as long as $\sum_{j=1}^{k_i} F(t_{i,j}^-)$ and $\sum_{j=1}^{k_i} F(t_{i,j}^+)$ are bounded away of zero. 

Under Assumption~\ref{assumption:detection-regions}, we have the following theorem, which provides
a local asymptotic minimax lower bound on the efficiency of
\emph{any} non-adaptive estimator.
\begin{theorem}
  \label{theorem:non-adaptive-minimax}
  Let Assumption~\ref{assumption:detection-regions} hold, and
  let ${\theta}_n$ be an estimator of $\theta \in \Theta$ from
  observations $B_i = \indic{X_i \in A_i}$.
  Then for $Z \sim \normal(0, 1)$ and any
  symmetric and quasi-convex function $L$,
  \begin{align*}
    & \liminf_{c \to \infty} \liminf_{n \to \infty}
    \sup_{\tau\,:\,|\theta-\tau| \leq \frac{c}{\sqrt{n} }}
    \E \left[ L\left( \sqrt{n}({\theta}_{n} - \tau) \right) \right] \\
    & \qquad \qquad \qquad \qquad \geq
    \E\left[ L (Z/\sqrt{\kappa(\theta)}) \right].
  \end{align*}
\end{theorem}
\noindent
See Appendix~\ref{sec:proof-non-adaptive-minimax} for a proof.

Theorem~\ref{theorem:non-adaptive-minimax} shows that the limiting variance
term $\kappa(\theta)$ provides a strong lower bound on the efficiency of any
non-adaptive estimator, and moreover, that this bound necessarily depends
on $\theta$. As a
particular consequence, for the squared error $L(x) = x^2$, for any $\delta
> 0$ and $\theta$, there exists a $c < \infty$ such that $\sup_{|\tau -
  \theta| \le c / \sqrt{n}} \E_\tau[(\theta_n - \tau)^2] \ge \frac{(1 -
  \delta)}{n \kappa(\theta)} + o(1/n)$. Consequently, attaining any type of good (uniform)
efficiency with non-adaptive estimators will be challenging. 

Yet, Theorem~\ref{theorem:non-adaptive-minimax}
limits non-adaptive strategies in stronger ways.
Under the density models we have considered, with the additional
Assumption~\ref{assump:failure_rate}, we can show stronger
optimality results that adaptivity is essential for achieving
optimal convergence guarantees.
Recall the transformation~\eqref{eq:eta_def} of the
hazard rate function, $\eta(x) = \frac{f^2(x)}{F(x)(1 - F(x))}$, which
has unique maximum at $x = 0$ under Assumption~\ref{assump:failure_rate}.
When each detection region $A_n$ consists of a bounded number
of intervals, the next theorem shows that
the minimal risk $1/\eta(0)$ can
only be attained at finitely many points within $\Theta$.
In particular, distinct from the adaptive setting, no distributed
estimation scheme can achieve asymptotic variance $\eta(0)$
uniformly in $\theta \in \Theta$.

\begin{thm} \label{thm:non_existence}
  Let Assumptions~\ref{assump:failure_rate}
  and~\ref{assumption:detection-regions} hold.
  Additionally, assume that $A_i$ is the union of at most $K$
  intervals. The number of points $\theta \in \Theta$ satisfying
  $\kappa(\theta) = \eta(0)$ is at most $2K$.
\end{thm}
\noindent
See Appendix~\ref{proof:thm:non_existence} for a proof.


\subsection{Threshold Detection}
\label{subsec:threshold}

We now consider a restricted case where each detection region is a
half-open interval, i.e., the $i$th message is obtained by comparing $X_i$
against a single threshold. Under the adaptive signal acquisition setting,
this is sufficient for asymptotic optimality;
in non-adaptive settings, it is not sufficient, though we may characterize
a few additional optimality results.
Assume now that each $B_i$ is of the form
\begin{equation}
  \label{eq:threshold_message}
  B_i = \sgn(t_i - X_i) = \begin{cases} 1 & X_i<  t_i, \\
    -1 & X_i \geq t_i,
  \end{cases}  
\end{equation}
where $t_i\in\mathbb R$ is the \emph{threshold} of the $i$th encoder. In
other words, the detection region of $B_i$ is $A_i = (t_i,\infty)$ and
$\mathbb P(X_i \in A_i) = F \left( B_i(t_i-\theta) \right)$. It follows that
\begin{align}
  L_n(A_1,\ldots,A_n;\theta)
  & = \frac{1}{n} \sum_{i=1}^n \frac{ \left(f(t_i-\theta) \right)^2 }{F\left(t_i-\theta \right) F\left(\theta - t_i \right) } \\
  & = \frac{1}{n} \sum_{i=1}^n \eta(t_i - \theta).
  \label{eq:Ln_threshold}
\end{align}
A natural condition for the existence of the limit \eqref{eq:Ln_threshold}
as $n\to \infty$ is that the empirical distribution of the threshold values
converges to a probability measure. Specifically, for an interval $I \subset
\mathbb R$, define
\begin{equation*}
  \lambda_n(I) = \frac{ \card \left( I \cap \{t_1,t_2,\ldots \} \right)}{n}. 
\end{equation*}
Then an investigation of the proof of
Theorem~\ref{theorem:non-adaptive-minimax} in
Section~\ref{sec:proof-non-adaptive-minimax}, specifically
Sec.~\ref{sec:proof-lan-bits} and the bounds~\eqref{eqn:h-fourth}, show that
as $\eta(t) \le \eta(0)$ for all $t \in \R$ under
Assumption~\ref{assump:failure_rate}, the following corollary follows. (The
corollary relies on local asymptotic normality~\cite[Ch.~7]{VanDerVaart98};
see Appendix~\ref{sec:proof-sgd-regular} for some brief discussion of such
conditions.)
\begin{cor} \label{cor:LAN_thresh}
  Let $\{t_n\}_{n=1}^\infty$ be a sequence of threshold values such that
  $\lambda_n$ converges (weakly) to a probability measure $\lambda$ on
  $\mathbb R$. Then the conclusions of
  Theorem~\ref{theorem:non-adaptive-minimax} apply with
  \begin{equation*}
    \kappa(\theta) = \int_{\mathbb R} \eta(t-\theta) \lambda(dt). 
  \end{equation*}
  Moreover, the family of laws of $\{B_i = \sgn(X_i - t_i)\}_{i = 1}^n$
  under $\{P_\theta\}_{\theta \in \Theta}$ is locally asymptotically normal
  with information $\kappa(\theta)$.
\end{cor}

The condition that $\lambda_n$ converges to a probability measure is
satisfied, for example, whenever $t_1,\ldots,t_n$ are drawn
independently from a probability distribution $\lambda(dt)$ on $\mathbb
R$.

When the conclusions of Corollary~\ref{cor:LAN_thresh} hold, local asymptotic normality of $\{B_n\}_{n=1}^\infty$ implies that the maximum
likelihood estimator (ML) of $\theta$ from $B_1,\ldots,B_n$, denoted here by
${\theta}^{ML}_n$, is local asymptotic minimax in the sense that
\begin{equation*}
  \sqrt{n} \left( {\theta}^{ML}_n - \theta \right)
  \cd \normal\left(0, 1/\kappa(\theta) \right). 
\end{equation*}
We note that ${\theta}^{ML}_n$ solves
\begin{equation}
  \label{eq:ML}
  0 = \sum_{i=1}^n B_i \frac{f \left( t_i-\theta\right) }{F \left(B_i  (t_i-\theta)\right) }.
\end{equation}
If the collection $\{t_1,t_2\ldots\}$ is bounded
(for example $\{t_1,t_2\ldots\} \subset \Theta$), then
\begin{equation*}
\lim_{n\to \infty} n \cdot \ex{\left({\theta}^{ML}_n - \theta \right)^2}  = 1/\kappa(\theta), 
\end{equation*} 
so that the ML estimator attains the local asymptotic MSE of Theorem~\ref{theorem:non-adaptive-minimax}.

By Assumption~\ref{assump:failure_rate},
$\eta(x)$ attains its maximum at the origin, so we conclude that
\begin{equation*}
  \kappa(\theta) \leq \sup_{t\in \mathbb R} \eta \left( t-\theta\right) = \eta(0).
\end{equation*}
Moreover, this upper bound on $\kappa(\theta)$ is attained only when
$\lambda$ is the point mass at $\theta$. Since $\theta$ is \emph{a priori}
unknown, estimation in the distributed setting using
threshold detection is strictly suboptimal compared to the adaptive
setting; the ability to choose the thresholds $t_i$
adaptively conditional on previous messages is necessary for optimal
efficiency.


\subsection{Minimax Threshold Density}

We conclude this section by considering the \newtext{distribution} of the threshold values
that maximizes the worst-case information $\inf_\theta \kappa(\theta)
= \kappa_\lambda(\theta)$ where $\kappa_\lambda(\theta)
= \int \eta(t - \theta) \lambda(dt)$.
The \newtext{optimal distribution $\lambda^\star$} solves the
optimization problem
\begin{align}
  \label{eq:var_cvx_minimax}
  \begin{split}
    \mathrm{maximize} \quad &  \inf_{\theta \in \Theta} \int \eta(t-\theta) \lambda(dt)
    \\ 
    \mathrm{subject~to} 
    \quad & \lambda(dt)\geq 0,\quad \int \lambda(dt) \leq 1. 
  \end{split}
\end{align}
The objective function~\eqref{eq:var_cvx_minimax} is concave in
$\lambda(dt)$ and continuous in the weak topology over measures on $\Theta$,
so that by discretizing, we can approximately solve this problem using
convex optimization. We let $\kappa^\star$ denote the maximal value of
problem~\eqref{eq:var_cvx_minimax} and $\lambda^\star(dt)$ be the density
achieving the maximum. By drawing thresholds
$t_i \simiid \lambda^\star$,
Corollary~\ref{cor:LAN_thresh} guarantees that for any $\theta \in \Theta$, the
maximum likelihood estimator
using $\{B_i = \sgn(X_i - t_i)\}_{i \in \N}$ is at least $\kappa^\star$. \par
Figure~\ref{fig:minimax_support} illustrates an approximation to
$\lambda^\star(dt)$ obtained by solving a discretized version of
\eqref{eq:var_cvx_minimax} for the case when $f(x)$ is the normal density
with variance $\sigma^2$ and $\Theta = [-1/2,1/2]$. The minimax asymptotic
precision parameter $\kappa^\star$ obtained this way is illustrated in
Fig.~\ref{fig:minimax_ARE} as a function of $\sigma$. Also
illustrated in these figures is $\kappa_{\unif}$, the precision
parameter corresponding to threshold values uniformly distribution over
$\Theta$,
\begin{align}
& \kappa_{\unif} \triangleq \min_{\theta \in [-T,T]} \frac{1}{2T}\int_{-T}^T \eta\left(t-\theta\right) dt \nonumber
 \\
& = 
\frac{1}{2T}\int_{-T}^{T} \eta\left(t\pm T\right) dt
= \frac{1}{2T}\int_{0}^{2T} \eta(t) dt  \label{eq:uniform_risk}. 
\end{align}

\begin{figure}
  \begin{center}
\begin{tikzpicture}[scale = 0.75]
\begin{axis}[
width=6cm, height=5cm,
xmin = -0.5, xmax=0.5, 
restrict y to domain = 0:100,
ymin = 0,
ymax = 1.05,
samples=10, 
xlabel= {$dt$},
xtick={-0.5,0,0.5},
xticklabels={$-.5$, 0, $.5$},
title = {$\sigma = .5$},
ytick={0,0.636},
ylabel = {$\lambda(dt)$},
yticklabels={0,\scriptsize $\frac{2}{\pi}$,1},
line width=1.0pt,
mark size=1.5pt,
ymajorgrids,
xmajorgrids,
legend style= {at={(1,1)},anchor=north east,draw=black,fill=white,align=left}
]
\addplot[color = blue, smooth, mark = o ] plot table [col sep=comma] {./Figs/minmax_lmd_b0.5_sig0.5.csv};

\addplot[color = black!30!green, smooth, dashed] plot table [x = x, y  = z,col sep=comma] {./Figs/minimax_th_b0.5_sig0.5.csv};

\addplot[color = red, smooth] plot table [x = x, y  = y,col sep=comma] {./Figs/minimax_th_b0.5_sig0.5.csv};

\end{axis}
\end{tikzpicture}
\begin{tikzpicture}[scale = 0.75]
\begin{axis}[
width=6cm, height=5cm,
xmin = -0.5, xmax=0.5, 
restrict y to domain = 0:100,
ymin = 0,
ymax = 0.7,
samples=10, 
xlabel= {$dt$},
title = {$\sigma = .2$},
xtick={-0.5,0,0.5},
xticklabels={$-.5$, 0, $.5$},
ytick={0,0.636},
yticklabels={0,\scriptsize $\frac{2}{\pi}$},
line width=1.0pt,
mark size=1.5pt,
ymajorgrids,
xmajorgrids,
legend style= {at={(1,1)},anchor=north east,draw=black,fill=white,align=left}
]
\addplot[color = blue, smooth, mark = o ] plot table [col sep=comma] {Figs/minmax_lmd_b0.5_sig0.2.csv};

\addplot[color = black!30!green, smooth, dashed] plot table [x = x, y  = z,col sep=comma] {Figs/minimax_th_b0.5_sig0.2.csv};

\addplot[color = red, smooth] plot table [x = x, y  = y,col sep=comma] {Figs/minimax_th_b0.5_sig0.2.csv};

\end{axis}
\end{tikzpicture}

\begin{tikzpicture}[scale = 0.75]
\begin{axis}[
width=6cm, height=5cm,
xmin = -0.5, xmax=0.5, 
restrict y to domain = 0:100,
ymin = 0,
ymax = 0.7,
samples=10, 
xlabel= {$dt$},
xtick={-0.5,0,0.5},
title = {$\sigma = .1$},
xticklabels={$-.5$, 0, $.5$},
ytick={0,0.636},
ylabel = {$\lambda(dt)$},
yticklabels={0,\scriptsize $\frac{2}{\pi}$},
line width=1.0pt,
mark size=1.5pt,
ymajorgrids,
xmajorgrids,
legend style= {at={(1,1)},anchor=north east,draw=black,fill=white,align=left}
]
\addplot[color = blue, smooth, mark = o ] plot table [col sep=comma] {Figs/minmax_lmd_b0.5_sig0.1.csv};

\addplot[color = black!30!green, smooth, dashed] plot table [x = x, y  = z,col sep=comma] {Figs/minimax_th_b0.5_sig0.1.csv};

\addplot[color = red, smooth] plot table [x = x, y  = y,col sep=comma] {Figs/minimax_th_b0.5_sig0.1.csv};

\end{axis}
\end{tikzpicture}
\begin{tikzpicture}[scale = .75]
\begin{axis}[
width=6cm, height=5cm,
xmin = -0.5, xmax=0.5, 
restrict y to domain = 0:10,
ymin = 0,
ymax = .7,
samples=10, 
xlabel= {$dt$},
title = {$\sigma = .05$},
xtick={-0.5,0,0.5},
xticklabels={$-.5$, 0, $.5$},
ytick={0,0.636},
yticklabels={0,\scriptsize $\frac{2}{\pi}$},
line width=1.0pt,
mark size=1.5pt,
ymajorgrids,
xmajorgrids,
legend style= {at={(1,1)},anchor=north east,draw=black,fill=white,align=left}
]
\addplot[color = blue, smooth, mark = o ] plot table [x = x, y  = y,col sep=comma] {Figs/minmax_lmd_b0.5_sig0.05.csv};

\addplot[color = black!30!green, smooth, dashed] plot table [x = x, y  = z,col sep=comma] {Figs/minimax_th_b0.5_sig0.05.csv};

\addplot[color = red, smooth] plot table [x = x, y  = y,col sep=comma] {Figs/minimax_th_b0.5_sig0.05.csv};

\end{axis}
\end{tikzpicture}
\caption{\label{fig:minimax_support}
Optimal threshold density under distributed encoding. The threshold density $\lambda^\star(dt)$ (blue) that maximizes the asymptotic relative efficiency for $f(x)$ the normal density with variance $\sigma^2$ and $\Theta= [-1/2,1/2]$. 
The continuous curve (red) is the ARE for each $\theta \in [-1/2,1/2]$ under the optimal density, hence the minimax ARE is the minimal value of this curve. The dashed curve (green) is the ARE when the threshold values are uniformly distributed over $[-1/2,1/2]$; its minimal value is $\kappa_{\unif}$ \eqref{eq:uniform_risk}. }
\end{center}
\end{figure}

\begin{figure}
  \begin{center}
\begin{tikzpicture}[scale = 1]
\begin{axis}[
width=8cm, height=6cm,
xmin=.03,
xmax=2.5, 
xmode = log,
restrict y to domain = 0:100,
ymin = 0,
ymax = 1,
samples=1, 
xlabel= {$\sigma$},
ytick={0,0.637,1},
yticklabels={0,$\frac{2}{\pi}$,1},
ylabel = {\scriptsize Relative Efficiency},
line width=1.0pt,
mark size=1.5pt,
ymajorgrids,
xmajorgrids,
legend style= {at={(1,1)},anchor=north east,draw=black,fill=white,align=left}
]
\addplot[color = red, smooth,
line width=1.0pt,
] plot table [x = x, y = y, col sep=comma] {./Figs/minmax_ARE_b0.5.csv};
\addlegendentry{\scriptsize optimal threshold density};

\addplot[color = black!35!green, dashed, line width=1.0pt,
        ]
 plot table [x = x, y = z, col sep=comma] {./Figs/minmax_ARE_b0.5.csv};
\addlegendentry{\scriptsize uniform threshold density};

\addplot[color = black, smooth, dotted, line width = 1pt] 
coordinates {
            (0.03, .637) (10, .637)
            };
\addlegendentry{\scriptsize attained in the adaptive case};
\end{axis}
\end{tikzpicture}
\caption{\label{fig:minimax_ARE} 
Minimax relative efficiency under distributed encoding. ARE versus $\sigma$ for $f(x)$ the standard normal density with variance $\sigma^2$ and parameter space $\Theta = [-1/2,1/2]$. The dashed curve (green) is the ARE under a uniform threshold density over $\Theta$ given by $K_{\unif}\sigma^2$ of \eqref{eq:uniform_risk}. \newtext{
The line $\pi/2$ is attained under adaptive encoding uniformly over the parameter space for any $\sigma$.}}
\end{center}
\end{figure}

\section{Conclusions \label{sec:conclusions}}
We considered the risk and efficiency in estimating the mean of a symmetric and log-concave distribution from a sequence of bits, where each bit is obtained by encoding a single sample from this distribution. 
In an adaptive encoding setting, we showed that, asymptotically, no estimator can be more efficient than the median of the samples. We also showed that this bound is tight by presenting two adaptive encoding and estimation procedures that are as efficient as the median. Furthermore, we showed that only one round of adaptivity is required to attain optimal efficiency. In the distributed setting we provided conditions for local asymptotic normality of the encoded samples, which implies asymptotic minimax bound on both the risk and efficiency relative to the mean. 
Under local asymptotic normality, the optimal estimation performance derived for the adaptive case can only be attained over a finite number of points, i.e., no scheme is uniformly optimal in this setting. 
We further considered the special case where the sequence of bits is obtained in a distributed manner by comparing against a prescribed sequence of thresholds. We characterized the performance of the optimal estimator from such bit-sequence using the density of the thresholds and considered the density that minimizes the minimax risk. 
%

Natural extensions of this work include situations when the communication bit-budget $b$ is larger than one and when each sample is a $d$-dimensional vector. Bounds on rate of convergence of the MSE in this general case follow from several recent works (e.g. \cite{zhang2013information,shamir2014fundamental,
braverman2016communication,han2018geometric,
barnes2020lower,cai2020distributed}), that in particular imply that in some cases the MSE decreases in the regular parametric rate of $1/n$ when $b$ and $d$ are held fixed in the sample size $n$. Nevertheless, the coefficient of the leading $1/n$ term corresponding to the ARE, which we characterized here in the case $b=1$ and $d=1$, is still unknown in the general case.


\newpage

\section*{Appendices}


\section{Fast convergence of uniform estimators under bit constraints}
\label{sec:uniform-weirdos}

Here we consider the uniform distribution as our location family,
demonstrating that in the adaptive setting~\eqref{item:adaptive} or even the
one-step adaptive setting~(\ref{item:one-step-adaptive}'), constrained
estimators can attain rates faster than the $1 / \sqrt{n}$ rates regular
estimands allow. Indeed, define $c(x) = -\log 2$ for $x \in [-1, 1]$ and
$c(x) = -\infty$ for $x \not \in [-1, 1]$. Then $f(x) = e^{-c(x)}$ is
log-concave and symmetric, and we may consider the location family with
densities $f(x - \theta)$. For notational simplicity, we assume we have
a sample of size $2n$. We provide a proof sketch that
there is a one-step adaptive estimator $\theta_n$ such that
\begin{equation}
  \label{eqn:uniforms-are-easy}
  \sup_{|\theta| \le \log n}
  P_\theta\left(|\theta_n - \theta| \ge \frac{\newtext{16} \log n}{n^{3/4}}\right)
  \le \frac{2}{n^2}.
\end{equation}
for all large $n$,
and so (by the Borel-Cantelli lemmas), for any $\theta \in \R$ we have
$P_\theta(|\theta_n - \theta| \le \newtext{16 \log n} / n^{3/4} ~
\mbox{eventually}) = 1$. This is of course faster than the $1/\sqrt{n}$
rates we prove throughout.

To prove inequality~\eqref{eqn:uniforms-are-easy}, we proceed in two steps,
both quite similar.
First, we define an initial estimator $\theta\init_n$.
Let $\epsilon > 0$, which we will determine presently, though we will
take $n \epsilon \to \infty$ as $n \to \infty$, so that we may assume
w.l.o.g.\ that $\theta \in [-n\epsilon/2, n \epsilon/2]$. Take the interval
$[-n\epsilon, n\epsilon]$, and construct
$m$ thresholds at intervals of size $2 n \epsilon / m$; let
the $j$th such threshold be
\begin{equation*}
  t_j \defeq -n \epsilon + \frac{2 n (j-1) \epsilon}{m}
\end{equation*}
Then we ``assign'' observations to each pair of thresholds, so that
threshold $j$ corresponds to observations $I_j \defeq \{\frac{n (j - 1)}{m}
+ 1, \ldots, \frac{n j}{m}\}$, of which there are $n/m$. For each index $i
\in I_j$, we set
\begin{equation*}
  B_i = \begin{cases}
    1 & \mbox{if~} X_i \newtext{- 1} \ge t_j \\
    0 & \mbox{otherwise}.
  \end{cases}
\end{equation*}
Then we simply set $\theta\init_n$ to be the \newtext{minimal} threshold for which
$B_i = 0$ for all observations $X_i$ corresponding to that threshold. \newtext{Denote by $j^*$ the index of the threshold corresponding to $\theta\init_n$.}

\newtext{Let us now consider the probability that $\theta_n\init$ is substantially wrong. Set $\theta_M \equiv \max_{i \in I_j^*} X_i - 1$. Note that we always have $\theta\init_n \ge \theta_M$ because no observations will be above
$t_j^*+1$, and that $\theta\init_n \le \theta_M + 2n\epsilon/m$. 
 In addition, 
\begin{align*}
P_\theta\left( \left| \theta_M - \theta \right| \ge \frac{2n\epsilon}{m} \right) = \left(1 - \frac{2 n \epsilon}{m} \right)^{n/m}.
\end{align*}
Putting it all together using the triangle inequality, we have
\begin{align*}
  P_\theta\left( \left| \theta\init_n - \theta \right| \ge \frac{4n\epsilon}{m} \right) \le \left(1 - \frac{2 n \epsilon}{m} \right)^{n/m} \le e^{-2 \frac{n^2}{m^2} \epsilon}
  .
\end{align*}
Therefore, setting the number of bins $m = \sqrt{n}$ and
the resolution $\epsilon = \log n / n$, 
\begin{align}
  \label{eqn:quality-of-initial-estimate}
  \sup_{|\theta| \le \log n} P_\theta\left(|\theta_n\init - \theta| \ge
  \frac{4 \log n}{\sqrt{n}}\right) \le \frac{1}{n^2}.
\end{align}
}

The second stage estimator follows roughly the same strategy, except that
the resolution of the bins is tighter. In particular, let us assume that
$|\theta_n\init - \theta| \le \frac{8 \log n}{\sqrt{n}}$, which happens
eventually by inequality~\eqref{eqn:quality-of-initial-estimate}.  (We will
assume this tacitly for the remainder of the argument.)  Consider the
interval $\Theta_n \defeq \theta_n\init + [-\frac{16 \log n}{\sqrt{n}},
  \frac{16 \log n}{\sqrt{n}}]$ centered at $\theta_n\init$; we know that the
interval includes $[\theta - \frac{8 \log n}{\sqrt{n}}, \theta + \frac{8
    \log n}{\sqrt{n}}]$. Without loss of generality we assume $\theta_n\init
= 0$.  Following precisely the same discretization strategy as that for
$\theta_n\init$, we divide $\Theta_n$ into $m$ equal intervals, with
thresholds $t_j = -\frac{16 \log n}{\sqrt{n}} + \frac{32 (j - 1) \log n}{m
  \sqrt{n}}$; let $\epsilon_n = \frac{32 \log n}{m \sqrt{n}}$ be the width of
these intervals.  Then following exactly the same reasoning as above, we
assign indices $I_j = \{\frac{n(j - 1)}{m} + 1, \ldots, \frac{n j}{m}\}$ and
for $i \in I_j$, set $B_i = 1$ if $X_i \newtext{- 1} \ge t_j$. We define $\theta_n$ to be
the \newtext{minimal} threshold $t_j$ for which $B_i = 0$ for all observations $X_i
\in I_j$. Then following precisely the reasoning above, we have (on the
event that $|\theta_n\init - \theta| \le \frac{8 \log n}{\sqrt{n}}$)
\begin{align*}
  & P_\theta(|\theta_n - \theta|
  \ge 2 \epsilon_n)
  \le (1 - \epsilon_n)^\frac{n}{m}
  \le \exp\left(-\frac{n \epsilon_n}{m}\right) \\
  & \quad = \exp\left(-\frac{32 \sqrt{n} \log n}{m^2}\right).
\end{align*}
Set $m = 4 n^{1/4}$ to obtain the claimed
result~\eqref{eqn:uniforms-are-easy}.

\subsection{Proof of Proposition~\ref{prop:CEO}
\label{app:proof:CEO}}

Denote by $D^\star$ the optimal MSE in the Gaussian CEO with $L$ observers and under a total sum-rate $r = r_1 + \ldots +r_L$. An expression for $D^\star$ as a function of $r$ is give as \cite[Eq. 10]{chen2004upper}:
\begin{equation} \label{eq:ceo_optimal_sumrate}
r = \frac{1}{2} \log^+ \left[ \frac{\sigma_\theta^2}{D^\star} \left( \frac{D^\star L}{ D^\star L - \sigma^2 + D^\star \sigma^2 / \sigma_\theta^2 }\right)^L  \right].
\end{equation}
For the special case where $r = n$ and $L=n$, we have
\begin{equation} \label{eq:ceo_optimal_sumrate2}
n = \frac{1}{2} \log_2 \left[ \frac{\sigma_\theta^2}{D^\star} \left(\frac{ D^\star n }{D^\star n - \sigma^2 + D^\star \sigma^2/\sigma_\theta^2 }  \right)^n  \right].
\end{equation}
Consider the distributed encoding setting (iii) in the case where $f(x) = \Ncal(0,\sigma^2)$ and the prior on $\Theta$ is $\pi = \Ncal(0,\sigma_\theta^2)$. The Gaussian CEO problem of \cite{viswanathan1997quadratic} with a unit bitrate $r_1=\ldots = r_n =1$ at each terminal and blocklength $k=1$ reduces to our distributed setting (iii). Since $D^\star$ satisfying \eqref{eq:ceo_optimal_sumrate2} describes the MSE in the CEO setting under an optimal allocation of the sum-rate $r = n$ among $n$ encoders, it provides a lower bound to the minimal MSE in estimating $\theta$ in the distributed setting. 
\newtext{By noting that $1/D^\star$ grows no faster than a polynomial in $n$ \cite{viswanathan1997quadratic}, we rely on the expansion
\begin{align*}
\left(\frac{\sigma_{\theta}^2}{D^\star}\right)^{1/n} 
= 1 + \frac{\log \left(\frac{\sigma_{\theta}^2}{D^\star}\right)}{n}+\frac{\log^2\left(\frac{\sigma_{\theta}^2}{D^\star}\right)}{2 n^2}+O\left(n^{-3}\right),
\end{align*}
to obtain that, in limit $n\rightarrow \infty$, \eqref{eq:ceo_optimal_sumrate2} behaves as 
\begin{align*}
D^\star = 
\frac{4 \sigma^2}{3 n}
+ \frac{16 \sigma^2}{9 n^2 \sigma_{\theta}^2 }-\frac{4 \sigma^2 \log \left(\frac{\sigma_{\theta}^2}{D^\star}\right)}{9 n^2}
+O(n^{-3}).
\end{align*}
}
This implies Proposition~\ref{prop:CEO}. 

\section{Proof of Theorem~\ref{thm:adpative_lower_bound}}
\label{proof:thm:adpative_lower_bound}

We begin with two technical lemmas.
\begin{lem} \label{lem:bound_intervals}
  Let $f$ be a log-concave and symmetric density function
  for which Assumption~\ref{assump:failure_rate} holds. For any $x_1 \geq
  \ldots \geq x_n \in \R$,
  \begin{align}
    & \frac{ \left| \sum_{k=1}^n (-1)^{k+1} f(x_k) \right|^2 }{
      \left( \sum_{k=1}^n (-1)^{k+1} F(x_k) \right) \left(1- \sum_{k=1}^n (-1)^{k+1} F(x_k) \right) } \nonumber
     \\ & \qquad \leq  4f(0)^2.
    \label{eq:lem_bound_intervals}
  \end{align}
\end{lem}
\begin{lem} \label{lem:fisher_bound}
  Let $X$ be a random variable with a symmetric, log-concave, and
  continuously differentiable density function $f(x)$ such that Assumption~\ref{assump:failure_rate} holds. For a Borel measurable set $A$, define
  \begin{equation*}
    B(x) \defeq \begin{cases} 1
      & \mbox{if} ~x \in A, \\
      -1 & \mbox{if}~ x \notin A.
    \end{cases}
  \end{equation*}
  The Fisher information of $B$ with respect to $\theta$ is bounded from above by $\eta(0)$.
\end{lem}

Lemma~\ref{lem:bound_intervals} is the special case $\delta = 0$ of
Lemma~\ref{lem:bound_intervals_delta} to come in
Section~\ref{sec:bound_intervals_delta}. We now prove
Lemma~\ref{lem:fisher_bound}.

\begin{proof-of-lemma}[\ref{lem:fisher_bound}]  
  We first note that in the special case where $f$ is a normal density, Lemma~\ref{lem:fisher_bound} follows from \cite[Thm.~3]{Barnes2018}. The proof below, valid for any log-concave symmetric density satisfying Assumption~\ref{assump:failure_rate}, is based on a different techique than that of \cite{Barnes2018}. \par
  Write the Fisher information of $B$ with respect to $\theta$ as
  \begin{align}
    I_\theta & =  \ex{ \left( \frac{d}{d\theta} \log P\left( B | \theta \right) \right)^2 |\theta } \nonumber \\
    & = \frac{ \left(\frac{d}{d\theta} P(B=1|\theta) \right)^2}{P(B=1| \theta)} + \frac{ \left(\frac{d}{d\theta} P(B=-1|\theta) \right)^2} {P(B=-1| \theta)} \nonumber \\
    & =  \frac{ \left( \frac{d}{d\theta} \int_A f \left( x-\theta\right)dx \right)^2} { P(B=1| \theta) } + \frac{ \left( \frac{d}{d\theta}\int_A f \left( x-\theta \right)dx \right)^2} { P(B=-1| \theta) } \nonumber \\ 
    & \overset{(a)}{=} \frac{ \left( - \int_A f' \left( x-\theta \right)dx \right)^2} { P(B=1| \theta) } + \frac{ \left(- \int_A f' \left( x-\theta \right)dx \right)^2} { P(B=-1| \theta) } \nonumber \\ 
    & = \frac{\left( \int_A f'\left( x-\theta \right) dx \right)^2 }{  P(B=1 | \theta) \left(1-P(B=1|\theta) \right)  }, \nonumber \\
    & = \frac{\left( \int_A f'\left( x-\theta \right) dx \right) \left( \int_A f'\left( x-\theta \right) dx \right)}{ \left( \int_A f \left( x-\theta \right) dx \right)  \left(1- \int_A f \left( x-\theta \right) dx \right) }, \label{eq:lem_fisher_bound_proof1}
  \end{align}
  where differentiation under the integral sign in $(a)$ is justified since $f$ is log-concave hence a.e.\ differentiable
  (cf.~\cite{Bertsekas73}) with a.e.\ derivative $f'(x)$. By regularity of the Lebesgue, for any $\epsilon>0$ there exists a finite number $k$ of disjoint open intervals $I_1,\ldots I_k$ such that
  \begin{equation*}
    \int_{A\setminus \cup_{j=1}^k I_j }  dx < \epsilon.
  \end{equation*}
  It follows that for any $\epsilon' > 0$, the set $A$ in
  \eqref{eq:lem_fisher_bound_proof1} can be replaced by a finite union of disjoint intervals without increasing $I_\theta$ by more than
  $\epsilon'$. Consequently, we may proceed assuming that
  $A$ is of the form
  \begin{equation*}
    A = \cup_{j=1}^k (t^+_j,t^-_j),
  \end{equation*}
  with $\infty \leq t^-_1 \leq \ldots \leq t^-_k$, $t^+_1 \leq \ldots \leq t^+_k \leq \infty$ and $t^-_j \leq t^+_j$ for $j=1,\ldots,k$. Under this assumption, 
  \begin{align*}
    P_{\theta}(B_n=1) & = \sum_{j=1}^k \left( F \left(t^+_j-\theta\right) -  F \left(t^-_j-\theta\right)  \right),
  \end{align*}
  so we may rewrite Eq.~\eqref{eq:lem_fisher_bound_proof1} as
  \begin{align*}
    I_\theta & = \frac { \left( \sum_{j=1}^{k} \left[ f \left(t^+_j-\theta \right) - f \left( t^-_j-\theta \right) \right] \right)^2 } 
    { \left( \sum_{j=1}^k \left[ F \left( t^+_j-\theta \right) - F \left( t^-_j-\theta \right) \right] \right) }  \nonumber \\
    & \times \frac {1} 
    {1- \left( \sum_{j=1}^k \left[ F \left(  t^+_j-\theta \right) - F \left( t^-_j-\theta \right) \right] \right) } 
  \end{align*}
  It follows from Lemma~\ref{lem:bound_intervals} that for any $\theta \in \R$ and any choice of the intervals' endpoints,
  \begin{equation*}
    I_\theta \le
    \max_{t \in \{t^\pm_1,\ldots,t^\pm_k\} } 4f(t)^2 \leq 4 f(0)^2. 
  \end{equation*}
\end{proof-of-lemma}

We now prove Theorem~\ref{thm:adpative_lower_bound}. Write the Fisher information of $B_1,\ldots,B_n$ with respect to $\theta$ as 
\begin{align*}
I_\theta(B_1,\ldots,B_n) = \sum_{i=1}^n I_\theta (B_i|B_1,\ldots,B_{i-1}),
\end{align*}
where $I_\theta (B_i|B_{i-1},\ldots,B_1)$ is the Fisher information of the distribution of $B_i$ given $B_1,\ldots,B_{i-1}$. By the definition of the adaptive setting, $P_{\theta}(B_i|B_1,\ldots,B_{i-1}) = P_{\theta}(X_i \in A_i)$ for some Borel measurable $A_i$. Consequently, Lemma~\ref{lem:fisher_bound} applies, leading to the bound
\[
 I_\theta (B_i|B_{i-1},\ldots,B_1) \leq 4f(0)^2,
\]
We conclude
\begin{align}
I_\theta(B_1,\ldots,B_n) \leq 4f(0)^2 n
\label{eq:fisher_information}.
\end{align}
The Van Trees inequality in the version of 
\cite{gill1995applications} holds under the regularity conditions on $\pi(\cdot)$, which implies
\begin{align*}
\ex{ \left( \theta_n - \theta \right)^2} &  \geq \frac{1}{ \E_{\pi} \left[ I_\theta(B_1,\ldots,B_n)\right] + I_0}.
\end{align*}
Combining the last display with \eqref{eq:fisher_information}, we get
\begin{align*}
\ex{ \left( \theta_n - \theta \right)^2} &  \geq \frac{1}{ 4f(0)^2 n + I_0}.
\end{align*}


\subsection{Isoperimetric Lemma}
\label{sec:bound_intervals_delta}

The following lemma is essential to the proofs
of Theorems~\ref{thm:adpative_lower_bound} and
\ref{thm:non_existence}.

\begin{lem}
  \label{lem:bound_intervals_delta}
  Let $f$ be a log-concave and symmetric density
  function. Let $\delta\geq 0$. Assume that the function
  \begin{equation*}
    \eta_\delta(x) \triangleq  \eta^{1+\delta}(x)/f^\delta(x)
    = \frac{  \left( f(x) \right)^{2+\delta}}{\left(F(x)(1-F(x))\right)^{1+\delta}}
  \end{equation*}
  is non-increasing in $|x|$. Then for any $x_1 \ge \ldots \ge x_n \in \R$,
  \begin{align}
    & \frac{ \left| \sum_{i=1}^n (-1)^{i+1} f(x_i) \right|^{2+\delta} } {\left|
      \sum_{i=1}^n (-1)^{i+1} F(x_i) \right|^{1+\delta} \left|1- \sum_{k=1}^n
      (-1)^{i+1} F(x_i) \right|^{1+\delta} } \nonumber \\
      & \qquad \leq \max_{i} \eta_\delta(x_i).
    \label{eq:lem_bound_intervals_delta}
  \end{align}
\end{lem}
In particular, 
\begin{align*}
& \frac{ \left| \sum_{i=1}^n (-1)^{i+1} f(x_i) \right|^{2+\delta} }
{\left| \sum_{i=1}^n (-1)^{i+1} F(x_i) \right|^{1+\delta} \left|1- \sum_{i=1}^n (-1)^{i+1} F(x_i) \right|^{1+\delta} } 
\\ 
& \qquad \leq \eta_\delta(0) = 4^{1+\delta} f^{2+\delta}(0).
\end{align*}

\begin{proof-of-lemma}[\ref{lem:bound_intervals_delta}]
  Denote 
  \begin{equation*}
    \delta_n(x_1,\ldots,x_n) \triangleq \sum_{i=1}^{n} s_i f(x_i),
  \end{equation*}
  \begin{equation*}
    \Delta_n(x_1,\ldots,x_n) \triangleq  \sum_{i=1}^n s_i F(x_i), 
  \end{equation*}
  where $s_i \triangleq (-1)^{i+1}$. We use induction on $n \in \N$ to show that 
  \begin{align}
    & \frac{\left| \delta_n(x_1,\ldots,x_n) \right|^{2+\delta}} 
         {\left|\Delta_n(x_1,\ldots,x_n)\left(1- \Delta_n(x_1,\ldots,x_n) \right) \right|^{1+\delta} } \nonumber \\
         & \qquad \leq \max_{i}\eta_{\delta}(x_i). \label{eq:lemm:interval_bounds:to_show}
  \end{align}
  Since 
  \begin{equation*}
    \eta_\delta(x) =  \frac{  \left|\delta_1(x) \right|^{2+\delta}}{\left|\Delta_1(x)
      (1-\Delta_1(x)) \right|^{1+\delta}}, 
  \end{equation*}
  The case $n=1$ is trivial.  
  Assume that \eqref{eq:lemm:interval_bounds:to_show} holds for all integers up to $n = N$ and for any $x_1 \ge \ldots \ge x_N$. Consider the case $n = N+1$. Let $i^*$ be the index such that $x_{i^*}$ has minimal absolute value among $x_1,\ldots,x_N$. The assumption on $\eta_\delta(x)$ implies that
  \begin{equation*}
    \eta_\delta(x_{i^*}) = \max_i \eta_\delta(x_i).
  \end{equation*}
  Since the LHS of \eqref{eq:lem_bound_intervals_delta} is invariant to a sign flip of all $x_1,\ldots,x_{N+1}$, we may assume that $x_{i^*}$ is positive without loss of generality. 
  Set $x^* = x_{i^*}$ and let $k=i^*-1$. Consider the function
  \begin{align}
    & g(y_1,\ldots,y_N) \triangleq g(y_1,\ldots,y_N|x^*,k) \label{eq:g_def} \\
    &  \triangleq  \frac{ \left| \delta_{N+1}(y_1,\ldots,y_k,x^*,y_{k+1}\ldots,y_N) \right|^{2+\delta}}{
     \left|\Delta_{N+1}(y_1,\ldots,y_k,x^*,y_{k+1}\ldots,y_N)\right|}
    \nonumber \\
    & \times \frac{1}{\left|1 -\Delta_{N+1}\left(y_1,\ldots,y_k,x^*,y_{k+1}\ldots,y_N\right) \right|^{1+\delta}} \nonumber
  \end{align}
  The LHS of \eqref{eq:lemm:interval_bounds:to_show} is obtained by taking $y_i=x_{k_i}$ where $k_i$ is the $i$th element in $\{1,\ldots,N+1\}\setminus \{i^*\}$. It is therefore enough to prove that 
  \begin{equation*}
    \max_{(y_1,\ldots,y_N) \in A_N(x^*,k) } g(y_1,\ldots,y_N) \leq \eta_{\delta}(x^*),
  \end{equation*}
  where 
  \begin{align*}
    A_N(x^*,k) & \triangleq \left\{ (y_1,\ldots,y_N) \in \R^N\, \right. \\
    & \left. \qquad  : \, y_1 \ge \ldots \ge y_k \ge x^* \ge -x^* \ge y_{k+1} \ldots \ge y_N
    \right\}.
  \end{align*}
  Since $f(x)$ is log-concave and symmetric, we may write $f(x) = e^{c(x)}$
  where $c(x)$ is concave, symmetric, and superdifferentiable on the interior
  of its domain with
  supergradient set $\partial c(x) = \{v \in \R \mid c(y) \le c(x) + v (y - x) ~
  \mbox{for~all}~ y\}$; $c$ is also
  differentiable a.e.\ with derivative
  \begin{equation*}
    c'(x) \triangleq \frac{f'(x)}{f(x)}
  \end{equation*}
  (when it exists), and we otherwise simply treat $f'(x) / f(x) = c'(x) \in
  \partial c(x)$ as an arbitrary element of the superdifferential.  The
  supergradient sets $\partial c(x)$ are increasing, in that $v_0 \in
  \partial c(x_0)$ and $v_1 \in \partial c(x_1)$ implies that $(v_1 - v_0)
  (x_1 - x_0) \le 0$.  We first prove the lemma under the assumption that
  $c$ is strictly concave, or, equivalently, that $v_i \in \partial c(x_i)$
  implies that $(v_1 - v_0)(x_1 - x_0) < 0$ whenever $x_1 \neq x_0$; that
  is, $c'$ is strictly decreasing.

  The maximal value of $g(y_1,\ldots,y_N)$ is attained for the same
  $(y_1,\ldots,y_N) \in A_N(x^*,k)$ that maximizes
  \begin{align*}
    & \log(g)(y_1,\ldots, y_N) \\
    &  =  (2+\delta) \log \left( \delta_N  \right)  - (1+\delta) \log \left( \Delta_N  \right) - (1+\delta) \log \left(1 - \Delta_N \right),
  \end{align*}
  where in the last display and henceforth we suppress the arguments
  $y_1,\ldots,y_k,x^*,y_{k+1},\ldots, y_N$ of the functions $\delta_N$ and
  $\Delta_N$. 
  Within the interior of $A_N(x^*,k)$, all three expressions in
  \eqref{eq:g_def} within an absolute value are positive. It follows that
  partial derivative of $\log(g)(y_1,\ldots,y_N)$ with respect to $y_i$ within
  the interior of $A_N(x^*,k)$ is
  \begin{equation*}
    \frac{\partial \log(g)}{\partial y_i} = \frac{(2+\delta) s_i
      f'(x_i)}{\delta_N } -\frac{(1+\delta) s_i f(x_i)}{\Delta_N } +
    \frac{(1+\delta)s_i f(y_i)}{1-\Delta_N }.
  \end{equation*}
  We conclude that the gradient of $\log(g)$ vanishes if and only if
  \begin{equation}
    \label{eq:gradient_zero}
    c'(y_i) = \frac{f'(y_i)}{f(y_i)} = \frac{1+\delta}{2+\delta} \frac{\delta_N}{2} \left( \frac{1}{\Delta_N} - \frac{1}{1-\Delta_N } \right),
  \end{equation}
  for $i=1,\ldots,N$.
  Since we assumed that $\partial c(x)$ is injective,
  equality~\eqref{eq:gradient_zero} holds if and only if $y_1 = \ldots
  = y_N$. In this case, $g(y_1,\ldots,y_N) = \eta_\delta(x^*)$ if $N$ is
  even. If $N$ is odd and $y_1 = \ldots = y_N > x^*$, then
  \begin{align*}
    & g(y_1,\ldots,y_N) \\
    & \quad = \frac{\left| f(y_1)-f(x^*)  \right|^{2+\delta}} { 
      \left|  F(y_1) - F(x^*) \right|^{1+ \delta} 
      \left| 1 - (F(y_1) -F(x^*)) \right|^{1+ \delta} } 
  \end{align*} 
  which is bounded from above by $\eta_\delta(x^*)$ by the induction hypothesis. The case where $N$ is odd and $-x^* \leq y_1 = \ldots = y_N$ is similar. 
  We now consider the possibility that the maximum of $g(y_1,\ldots,y_N)$ is attained at the boundaries of $A_N(x^*,k)$. At boundary points for which $y_i = y_{i+1}$ for some $i$, the contribution of $y_i$ and $y_{i+1}$ to $g(y_1,\ldots,y_N)$ is zero and the induction assumption for $n=N-1$ implies that 
  \begin{equation*}
    g(y_1,\ldots,y_N) \leq \eta_{\delta}(x^*).
  \end{equation*}
  The remaining boundary points of $A_N(x^*,k)$ are covered by the following cases:
  \begin{itemize}
  \item[(i)]  $y_N = -\infty$. 
  \item[(ii)] $y_1 = \infty$.
  \item[(iii)] $y_k = x^*$.
  \item[(iv)] $y_{k+1} = -x^*$. 
  \end{itemize}
  For case (i), 
  \begin{align*}
    & g(y_1,\ldots,y_N) \\
    &  \to \frac{ \left| \sum_{i=1}^{k} s_i f(y_i) + s_{i^*} f(x^*) - \sum_{i=k+1}^{N-1} s_i f(y_i) \right|^{2+\delta}} 
    {\left| \sum_{i=1}^{k} s_i F(y_i) + s_{i^*} F(x^*) - \sum_{i=k+1}^{N-1} s_i F(y_i) \right|^{1+\delta}} \\
    & \medmath{ \times \frac{1}{\left|1- \sum_{i=1}^{k} s_i F(y_i) - s_{i^*} F(x^*) + \sum_{i=k+1}^{N-1} s_i F(y_i) \right|^{1+\delta}}},
  \end{align*}
  which is smaller than $\eta_\delta(x^*)$ by the induction hypothesis. Similarly, under case (ii),
  \begin{align*}
    & g(y_1,\ldots,y_N) \\
    & \to \medmath{
    \frac{ \left| \sum_{i=2}^{k} s_i f(y_i) + s_{i^*} f(x^*) - \sum_{i=k+1}^{N} s_i f(y_i) \right|^{2+\delta}} 
         {\left| 1+ \sum_{i=2}^{k} s_i F(y_i) + s_{i^*} F(x^*) - \sum_{i=k+1}^{N} s_i F(y_i) \right|^{1+\delta}} }
         \\
         & \medmath{ \times \frac{1}{
         \left|-\left( \sum_{i=2}^{k} s_i F(y_i) + s_{i^*} F(x^*) - \sum_{i=k+1}^{N} s_i F(y_i) \right)  \right|^{1+\delta}}}
         \\
         & ~ \medmath{ = \frac{ \left| -\sum_{i=2}^{k} s_i f(y_i) - s_{i^*} f(x^*) + \sum_{i=k+1}^{N} s_i f(y_i) \right|^{2+\delta}} 
           { \left|1 - \left(-\sum_{i=2}^{k} s_i F(y_i) - s_{i^*} F(x^*) + \sum_{i=k+1}^{N} s_i F(y_i) \right) \right|^{1+\delta} } }\\
          & \medmath{ \times \frac{1}{ 
             \left|-\sum_{i=2}^{k} s_i F(y_i) - s_{i^*} F(x^*) + \sum_{i=k+1}^{N} s_i F(y_i) \right|^{1+\delta} } },
  \end{align*}
  which is smaller than $\eta_{\delta}(x^*)$ by the induction hypothesis. Under case (iii), the terms in $\delta_N$ and $\Delta_N$ corresponding to $y_k$ and $x^*$ cancel each other. As a result,  $g(y_1,\ldots,y_N)$ reduces to an expression with $N-2$ variables hence this case is handled by the induction hypothesis. 
  Finally, under case (iv), set 
  \begin{equation*}
    d \triangleq s_k F(-x^*) + s_{i^*} F(x^*) = s_{i^*}\left(1-2F(-x^*) \right), 
  \end{equation*}
  \begin{equation*}
    \sigma \triangleq \sum_{i=1}^{k-1} s_i f(y_i) - \sum_{i=k+1}^{N} s_i f(y_i). 
  \end{equation*}
  and 
  \begin{equation*}
    \Sigma \triangleq \sum_{i=1}^{k-1} s_i F(y_i) - \sum_{i=k+1}^{N} s_i F(y_i). 
  \end{equation*}
  We have
  \begin{align*}
    & g(y_1,\ldots,y_N) =\nonumber  \\ 
    & =
    \frac{ \left| \sum_{i=1}^{k-1} s_i f(y_i) - \sum_{i=k+1}^{N} s_i f(y_i) \right|^{2+\delta}} 
         {\left| \sum_{i=1}^{k-1} s_i F(y_i) + d(x^*) - \sum_{i=k+1}^{N} s_i F(y_i) \right|^{1+\delta}} \\
         & \frac{1}{\left|1- \sum_{i=1}^{k-1} s_i F(y_i) - d(x^*) + \sum_{i=k+1}^{N} s_i F(y_i)   \right|^{1+\delta} }, \\
         & = \frac{ \left| \sigma \right|^{2+\delta}} 
         {\left| \Sigma+d \right|^{1+\delta} \left|1- \Sigma - d \right|^{1+\delta} } \\
         & = \frac{ \left| \sigma \right|^{2+\delta}} 
              {\left| \Sigma \right|^{1+\delta} \left|1- \Sigma \right|^{1+\delta} }  \left| \frac{\Sigma(1-\Sigma) } { \Sigma(1-\Sigma) + d(1-2\Sigma)-d^2} \right|^{1+\delta}. 
  \end{align*}
  By the induction hypothesis,
  \begin{equation*}
    \frac{ \left| \sigma \right|^{2+\delta}} 
         {\left| \Sigma \right|^{1+\delta} \left|1- \Sigma \right|^{1+\delta} } \leq \eta_\delta(x^*), 
  \end{equation*}
  hence it is left to show that 
  \begin{equation*}
    \frac{\Sigma(1-\Sigma) } { \Sigma(1-\Sigma) + d(1-2\Sigma)-d^2} \leq 1.
  \end{equation*}
  Whenever $d>0$, 
  \begin{align*}
    & \frac{ \Sigma(1-\Sigma) + d(1-2\Sigma)-d^2} {\Sigma(1-\Sigma)} \geq 1 \Leftrightarrow  1-2\Sigma \geq d, 
  \end{align*}
  while for $d<0$,
  \begin{align*}
    & \frac{ \Sigma(1-\Sigma) + d(1-2\Sigma)-d^2} {\Sigma(1-\Sigma)} \geq 1 \Leftrightarrow  1-2\Sigma \leq d. 
  \end{align*}
  Therefore, it is enough to show that $\Sigma \leq F(-x^*)$ if $s_{i^*}=1$ and 
  $\Sigma \geq F(-x^*)$ if $s_{i^*}=-1$. 
  Indeed, if $s_{i^*}=1$, then $s_{k+1}=-1$ and monotonicity of $F(x)$ implies that 
  \begin{equation*}
    \Sigma + d \leq F(y_1) - F(y_k) + F(x^*) - F(-x^*) + F(y_{k+2}) - F(y_N), 
  \end{equation*}
  and hence
  \begin{equation*}
    \Sigma \leq 1-F(x^*) = F(-x^*). 
  \end{equation*}
  Similarly, if $s_{i^*}=-1$ then 
  \begin{equation*}
    1 - \Sigma \leq 1 -  F(-x^*).
  \end{equation*}
  This conclude the proof in the case where $c'(x)$ is an injection. 

  In the case where $c(x)$ is not strictly concave, so that $c'$ does not
  strictly decrease, we approximate $c$ using another concave symmetric
  function with decreasing derivative. We assume w.l.o.g.\ that $c(0) = 0$
  maximizes $c$.  For $\alpha>0$ consider the function $f_\alpha(x) =
  \kappa(\alpha) e^{-|c(x)|^{1+\alpha}}$, where $\kappa(\alpha)$ normalizes
  $f_\alpha$.  Then $c_\alpha(x)$ is concave, symmetric, and
  a.e.\ differentiable with
  \begin{equation*}
    c_\alpha'(x) \triangleq \frac{f'_\alpha(x)}{f_\alpha(x)} = (1+\alpha)|c(x)|^{\alpha} c'(x). 
  \end{equation*}
  Now $c_\alpha'(x)$ is non-increasing since it is the derivative of a concave
  function. Furthermore, since $c(x)$ is non-constant on any interval and
  $c'(x)$ is non-increasing, $c_\alpha'(x)$ is non-constant on any interval
  hence an injection. It follows from the first part of the proof that, for
  any $\alpha>0$,
  \begin{align}
    \label{eq:proof:lem:bound_intervals}
    \frac{(\delta_{n,\alpha})^{2+\delta}}{\left(\Delta_{n,\alpha}(1-\Delta_{n,\alpha})\right)^{1+\delta}} \leq \max_i \eta_{\delta,\alpha}(x_i),
  \end{align}
  where 
  \begin{equation*}
    \delta_{n,\alpha} \triangleq  \sum_{k=1}^{n} (-1)^{k+1} f_{\alpha}(x_k),
  \end{equation*}
  \begin{equation*}
    \Delta_{n,\alpha} \triangleq \sum_{k=1}^n (-1)^{k+1} F_{\alpha}(x_k), 
  \end{equation*}
  and 
  \begin{equation*}
    \eta_{\delta,\alpha}(x) \triangleq \frac{(f_\alpha(x))^{2+\delta}}{\left(F_{\alpha}(x)(1-F(x)) \right)^{1+\delta}}. 
  \end{equation*}
  The proof is completed by noting that 
  \begin{align*}
    \lim_{\alpha \to 0} \frac{(\delta_{n,\alpha})^{2+\delta} }{ \left(\Delta_{n,\alpha}(1-\Delta_{n,\alpha})\right)^{1+\delta}}  = \frac{(\delta_{n})^{2+\delta }}{\left(\Delta_{n}(1-\Delta_{n}) \right)^{1+\delta}},  
  \end{align*}
  and, since the maximum is over a finite set,
  \begin{align*}
    \lim_{\alpha \to 0}  \max_i \eta_{\delta,\alpha}(x_i)  = \max_i\eta_\delta(x_i).
  \end{align*}


\end{proof-of-lemma}

\section{Proof of Theorem~\ref{thm:sgd}}
\label{proof:sgd}

The estimation algorithm~\eqref{eq:sgd_alg} is a special
case of the stochastic gradient procedures in the papers
\cite{polyak1992acceleration, polyak1990new}.
We rely on several of their results. Throughout this proof,
we assume without loss of generality that the median
$\theta = \mbox{med}(P) = 0$.

\subsection{Proof of Theorem~\ref{thm:sgd}\eqref{item:normal-sgd}}
\label{sec:proof-normal-sgd}

Consider the following simplified version of
\cite[Thm. 4]{polyak1992acceleration}:
\begin{corollary}{\cite[Thms. 3 \& 4]{polyak1992acceleration}}
  \label{corollary:polyak-juditsky}
  Let $\varphi : \R \to \R$ and $\{Z_i\}$ be i.i.d.\ zero-mean random
  variables, and
  \begin{equation*}
  X_i = \theta + Z_i.
  \end{equation*}
  Define
  \begin{align}
    \begin{split}
      \theta_i & = \theta_{i-1} + \gamma_i \varphi(X_i - \theta_{i-1}), \\
      \bar{\theta}_n & = \frac{1}{n} \sum_{i=0}^{n-1} \theta_i, 
    \end{split}
    \label{eq:Polyak_Juditsky_alg}
  \end{align}
  where in addition,
  \begin{enumerate}[(i)]
  \item There exists $K_1$ such that $\left| \varphi(x) \right| \leq
    K_1(1+|x|)$ for all $x\in \R$.
  \item The sequence $\left\{ \gamma_i \right\}_{i=1}^\infty$ satisfies
    condition~\eqref{eqn:lazy-gamma}.
  \item \label{item:zero-gradient}
    The function $\psi(x) \defeq \ex{ \varphi(x+Z_1)}$
    satisfies $\psi(0) = 0$ and
    $x\psi(x) > 0$ for $x\neq 0$. Moreover, $\psi$ is differentiable
    at 0 with $\psi'(0) > 0$ and there exists
    $K_2$, $0 < \lambda \leq 1$, \newtext{and $r>0$}, such that
    \begin{equation}
      \label{eqn:local-hessian-psi}
        \left| \psi(x) - \psi'(0)x \right|\leq K_2 |x|^{1+\lambda}
    \end{equation}
    \newtext{for all $|x|<r$. }
  \item The function 
    $\chi(x) \defeq \ex{\varphi^2(x+Z_1)}$ is continuous at zero. 
  \end{enumerate}
  Then $\bar{\theta}_n \cas \theta$ and $ \sqrt{n}({\theta}_n - \theta)
  \cd \normal(0,V)$ for
  $V = \frac{ \chi(0)} {\psi'(0)^2}$.
\end{corollary}

Using the notation in Corollary~\ref{corollary:polyak-juditsky}, we set
$\varphi(x) = \sgn(x)$ and $Z_i = X_i - \theta$, where $\theta =
\mbox{med}(P)$. Without loss of generality and for notational convenience,
we assume for the remainder of this derivation that $\theta = 0$.
As a consequence, we have $\mbox{med}(Z) = 0$,
and $\chi(x) = \ex{ \sgn^2(x+Z_1) }= 1$, so
$\chi(0) = 1$. In addition,
\begin{align*}
  \psi(x) & = \ex{ \sgn(x+ Z_1) } = 
  P(Z \ge -x) - P(Z < -x) \\
  & = 1 - 2 P(Z \le -x).
\end{align*}
Using that $P$ has a density $f$ near its median, it follows that $\psi'(x)
= 2f(-x)$ and thus $\psi'(0) = 2f(0) > 0$.  We may now verify that the
conditions in Corollary~\ref{corollary:polyak-juditsky} hold for $\lambda =
1$. Condition~(i) is obvious, and the convexity of $|\cdot|$ gives most of
condition~(iii) excepting inequality~\eqref{eqn:local-hessian-psi}. For
that, note that as $f$ is Lipschitz near $0$ with constant $\lip_0(f)$, we 
have for small $x$ that
\begin{align*}
  \psi(x) & = 2 \int_0^x f(-t) dt
  \leq 2 \int_0^x \left[f(0) + \lip_0(f) t \right] dt \\
  & = 2 f(0) x + \lip_0(f) x^2 = \psi'(0) x + \lip_0(f) x^2,
\end{align*}
\begin{align*}
  \psi(x) & = 2 \int_0^x f(-t) dt
  \geq 2 \int_0^x \left[f(0) - \lip_0(f) t \right] dt
  \\
  & = 2 f(0) x - \lip_0(f) x^2 = \psi'(0) x - \lip_0(f) x^2,
\end{align*}
so that condition~(iii) holds.
As evidently $\chi(0) / \psi'(0)^2 = \frac{1}{4 f(0)^2}$,
Corollary~\ref{corollary:polyak-juditsky} gives
Theorem~\ref{thm:sgd}\eqref{item:normal-sgd}.

\subsection{Proof of Theorem~\ref{thm:sgd}\eqref{item:sgd-regular}}
\label{sec:proof-sgd-regular}

This proof requires somewhat more technicality than the first part of the
theorem, including a brief detour into local asymptotic normality theory,
regular estimators, and quadratic-mean
differentiability~\cite[cf.]{VanDerVaart98}.
We assume without loss of generality that the median of the density
$f$ is 0, so that if $P_\theta$ has density $f(\cdot - \theta)$, the
median of $P_\theta$ is $\theta$.
We begin by recalling the statistical concepts we require.
\begin{definition}
  \label{definition:regular-estimator}
  A sequence of estimators $T_n$ for a parameter $\theta$ in the parametric
  family $\{P_\theta\}_{\theta \in \Theta}$ is \emph{regular at $\theta$} if
  there exists a distribution $Q$ such that for any bounded sequence
  $h_n$,
  \begin{equation*}
    \sqrt{n}(T_n - (\theta + h_n / \sqrt{n}))
    \mathop{\cd}_{P_{\theta + h_n/\sqrt{n}}}
    Q.
  \end{equation*}
\end{definition}
\begin{definition}
  \label{definition:qmd}
  Let $\{P_\theta\}_{\theta \in \Theta}$ have densities $p_\theta$
  with respect to a base measure $\mu$. The family
  is \emph{quadratic mean differentiable} (QMD) at
  $\theta$ with score $\score_\theta$ if
  \begin{equation}
    \label{eqn:qmd}
    \int \left(\sqrt{p_{\theta + h}} - \sqrt{p_\theta}
    - \half h^\top \score_\theta\sqrt{p_\theta} \right)^2 d\mu
    = o(\norm{h}^2)
  \end{equation}
  as $h \to 0$.
\end{definition}
\begin{definition}
  \label{definition:lan}
  A family of distributions $\{P_\theta\}_{\theta \in \Theta}$ is
  \emph{locally asymptotically normal with information matrix $I_\theta$} (LAN)
  at
  $\theta$ if there exists a sequence of random vectors $\{Z_n\}$ such that for
  all $h_n \to h$,
  \begin{equation*}
    \sum_{i = 1}^n \log \frac{dP_{\theta + h_n/\sqrt{n}}}{dP_\theta}
    \newtext{(X_i)}
    = h^\top Z_n - \half h^\top I_\theta h + o_P(1)
  \end{equation*}
  where $Z_n \cd \normal(0, I_\theta)$ under $P_\theta$, where
  $X_i \simiid P_\theta$.
\end{definition}

These three definitions are linked in our case by a few important results.
First~\cite[Theorem 7.2]{VanDerVaart98}, if $\{P_\theta\}$ is QMD
(Def.~\ref{definition:qmd}) at the point $\theta$, then it is locally
asymptotically normal with $Z_n = \frac{1}{\sqrt{n}} \sum_{i = 1}^n
\score_\theta(X_i)$ and information matrix $I_\theta =
\E_\theta[\score_\theta \score_\theta^\top]$. Moreover, in any family
$\{P_\theta\}$ that is LAN (Def.~\ref{definition:lan}) at $\theta$, if $T_n$
is a regular estimator (Def.~\ref{definition:regular-estimator}) at $\theta$
with limiting distribution $Q$, then for any bounded, symmetric,
quasi-convex loss $L$ and $c < \infty$,
\begin{align}
  \limsup_n \sup_{\norm{h} \le c}
  & \E_{P_{\theta + h/\sqrt{n}}}
  \left[L(\sqrt{n}(T_n - \theta - h / \sqrt{n}))\right] \nonumber \\
  & = \E[L(W)] ~~ \mbox{for~}W \sim Q \label{eqn:limit-law-regular}
\end{align}
(see Beran~\cite{beran1995role}, Eq.~(4.2)).
Thus, we show two results: first,
that the family $\{P_\theta\}$ of distributions
defined by the shifted densities $\{f(\cdot - \theta)\}_{\theta \in \R}$
is quadratic-mean-differentiable at any $\theta$,
and second, that $\bar{\theta}_n$ is regular and
asymptotically normal.
The combination evidently gives the theorem.

For quadratic mean differentiability, we have the following lemma, somewhat
more general than we need; we defer proof to Sec.~\ref{sec:proof-qmd}.
\begin{lemma}[Extension of \cite{VanDerVaart98}, Lemma 7.6]
  \label{lemma:qmd}
  Let $p_\theta$ be a density with respect to $\mu$,
  and assume that $\theta \mapsto s_\theta(x) \defeq \sqrt{p_\theta(x)}$
  is absolutely continuous for all $x$. Let
  $\dot{p}_\theta(x) = \nabla_\theta p_\theta(x)$ (when it exists),
  and assume that
  \begin{equation*}
    \mu(\{x : \dot{p}_\theta(x) ~ \mbox{fails~to~exist}\}) = 0.
  \end{equation*}
  Assume that $I_\theta \defeq \E_{P_\theta}[\dot{p}_\theta
    \dot{p}_\theta^\top / p_\theta^2]$ is continuous at $\theta_0$. Then
  $P_\theta$ is QMD (Definition~\ref{definition:qmd}) at $\theta = \theta_0$
  with $\score_\theta = \dot{p}_\theta / p_\theta$.
\end{lemma}

By the assumption in Theorem~\ref{thm:sgd} that the density $f$
is Lipschitz continuous, \newtext{$f$ is absolutely continuous hence 
$\sqrt{f}$ is absolutely continuous}. We see that the location family
$\{P_\theta\}_{\theta \in \R}$ defined by $dP_\theta(x) = f(x - \theta)$
satisfies the conditions of Lemma~\ref{lemma:qmd}.

It remains to show that the average $\bar{\theta}_n$ is regular \newtext{at $\theta$ with the limiting distribution $\Ncal(0,(4f(0)^2)^{-1})$}:
\begin{lemma}
  \label{lemma:sgd-median-regular}
  Let $h_n \to h \in \R$, and define $P_{n,h} = P_{\theta + h_n / \sqrt{n}}^n$.
  Then 
  \begin{align}
    \label{eqn:sgd-median-regular}
    \newtext{
    \sqrt{n}\left( \bar{\theta}_n - \theta \right)
    \mathop{\cd}_{P_{n,h}}
    \normal\left( h,\frac{1}{4 f(0)^2}\right).
    }
  \end{align}
\end{lemma}
\begin{proof}
  To show the convergence~\eqref{eqn:sgd-median-regular} we use the
  following refinement of Corollary~\ref{corollary:polyak-juditsky}, which
  provides a generalized convergence result for iteratively defined
  $\theta_n$, and whose
  proof we defer to Section~\ref{proof:normal-expansion}.
  \begin{corollary}
    \label{corollary:normal-expansion}
    Let the conditions of Corollary~\ref{corollary:polyak-juditsky} hold,
    meaning that $\theta_i = \theta_{i-1} + \gamma_i \varphi(X_i -
    \theta_{i-1})$ for $X_i = \theta + Z_i$, where $\{Z_i\}$ are i.i.d.\ with
   $\E[Z_1]=0$ and  $\E[\varphi(Z_1)] = 0$. Additionally assume the local smoothness
    condition that there exist $0 < \lambda \le 1$ and $K < \infty$ such
    that
    \begin{equation}
      \label{eqn:additional-local-smooth} 
      \E[|\varphi(x + Z_1) - \varphi(Z_1)|^2]
      \le K (|x|^\lambda + x^2).
    \end{equation}
    Set $\Delta_i \defeq \theta_i - \theta$ and $\bar{\Delta}_n \defeq \frac{1}{n}
    \sum_{i=1}^n \Delta_i$. Then
    \begin{enumerate}[(i)]
    \item \label{item:regularity}
      The sequence $\{\Delta_i\}$ is regular, that is,
      \begin{equation}
        \sqrt{n} \bar{\Delta}_n
        = -\frac{1}{\sqrt{n}} \frac{1}{\psi'(0)} \sum_{i=1}^{n-1} \varphi(Z_i)
        + o_{P,n}(1).
        \label{eq:normal_expansion_lem}
      \end{equation}
    \item \label{item:apply-le-cam} Let $\{Z_i\}$ as in
      Corollary~\ref{corollary:polyak-juditsky} have absolutely continuous
      density $p$ with median $0$, define $\score_h(z) = \frac{p'(z - h)}{p(z
        - h)}$, and assume that $I_h \defeq \E_p[\score_h(Z_1)^2]$ is
      continuous in $h$ near 0.  Then for any converging sequence $h_n \to h$,
      \begin{equation*}
        \sqrt{n} \bar{\Delta}_n
        \mathop{\cd}_{P_{\theta + h_n/\sqrt{n}}^n}
        \normal\left( \frac{-h}{\psi'(0)} \E_p[\varphi(Z_1) \score_0(Z_1)],
        \frac{\chi(0)}{\psi'^2(0)} \right).
      \end{equation*}
    \end{enumerate}
  \end{corollary}

  We now verify that the setting of Theorem~\ref{thm:sgd}
  (and Lemma~\ref{lemma:sgd-median-regular}) satisfies the
  conditions of Corollary~\ref{corollary:normal-expansion}. First, we have
  the obvious fact that
  \begin{equation*}
    |\sgn(z) - \sgn(x + z)|
    \le 2 \cdot \indic{|x| \ge |z|}.
  \end{equation*}
  Recalling that the density $f$ is Lipschitz with median 0,
  for $\varphi(z) = \sgn(z)$, and $Z = X - \theta$ distributed with
  density $f$, we have
  \begin{align*}
    & \ex{ \left| \varphi(Z) - \varphi(x + Z) \right| }
    \leq 2 \Prob\left( |Z_1| \le |x|  \right) \\
    & \quad = 2 \int_{-|x|}^{|x|} f(t) dt
    \le 4 f(0)|x| + 2 \int_{-|x|}^{|x|}
    \lip(f) t dt \\
     & \quad = 4 f(0) |x| + 2 \lip(f) x^2
  \end{align*}
  where $\lip(f)$ is the Lipschitz constant of $f$.
\newtext{It follows that condition~\eqref{eqn:additional-local-smooth} holds.}  
  In addition, we have
  \begin{align*}
& \newtext{ \E_p[\varphi(Z_1) \score_0(Z_1)] } = 
    \int_{\R} \varphi(x) f'(x ) dx
    = \int_{\R} \sgn(x) f'(x ) dx\\
     & \quad  =
    \int_0^\infty f'(x ) dx
    - \int_{-\infty}^0 f'(x ) dx
    = -2 f(0) = -\psi'(0).
  \end{align*}
  Corollary~\ref{corollary:normal-expansion}
  now implies the convergence \eqref{eqn:sgd-median-regular}.
\end{proof}

Combining Lemmas~\ref{lemma:qmd} and~\ref{lemma:sgd-median-regular}
with the limit~\eqref{eqn:limit-law-regular} gives
Theorem~\ref{thm:sgd}\eqref{item:sgd-regular}.

\subsection{Proof of Theorem~\ref{thm:sgd}\eqref{item:sgd-ms-convergence}}
\label{sec:proof-sgd-ms-convergence}

We begin with the following result from \cite{polyak1990new}:
\begin{corollary}[\cite{polyak1990new}, Theorem 2]
  \label{corollary:polyak-mse}
  Define the iteration
  \begin{align} \label{eq:polyak_new_measurements}
    \begin{cases}
      U_n = U_{n-1} - \gamma_n \varphi(Y_n), & Y_n = g'(U_{n-1})+Z_n \\
      \bar{U}_n= \frac{1}{n} \sum_{i=1}^n U_n, & n=1,2,\ldots.
    \end{cases}
  \end{align}
  Assume that the function $g$ is $\mc{C}^2$, strictly convex, has Lipschitz
  derivative, and is minimized by $x^\star$. Moreover, assume that the
  noises $\{Z_n\}$ are i.i.d.\ with density $p$ and
  that the Fisher information $\E[(p'(Z_1))^2 / p(Z_1)^2]$ exists and is finite.
  Let $\psi(x)$ and $\chi(x)$ be defined as in
  Corollary~\ref{corollary:polyak-juditsky} and satisfy the conditions in
  the corollary. Assume in addition that $\chi(0)>0$, condition
  \eqref{eqn:local-hessian-psi} with $\lambda = 1$, and there exits $K_3$
  such that
  \begin{equation*}
    \ex{  | \varphi(x+Z_1) |^4 } \leq K_3(1+|x|^4). 
  \end{equation*}
  Finally, assume that the sequence $\{\gamma_n \}$ satisfies conditions
  \eqref{eqn:lazy-gamma} and \eqref{eqn:stringent-gamma}. Then
  \begin{equation*}
    V_n \defeq
    \E\Big[(\bar{U}_n-x^\star )^2\Big] = n^{-1}\frac{\chi(0)} {
      (\psi'(0))^2 (g''(x^\star))^2 } + o(n^{-1}).
  \end{equation*}
\end{corollary}

\newtext{Fix $\theta \in \R$}. Apply Corollary~\ref{corollary:polyak-mse} with $g(x) = 0.5(x-\theta)^2$,
$\varphi(x) = \sgn(x)$, $Z_n = \theta-X_n$.
The update~\eqref{eq:polyak_new_measurements} gives
\begin{align*} 
  U_n 
  & = U_{n-1} + \gamma_n \sgn(X_n-U_{n-1} ),
\end{align*}
so the estimator $\bar{U}_n$ is identical to the stochastic gradient
estimator~\eqref{eq:sgd_alg} with $\bar{\theta}_n = \frac{1}{n} \sum_{i =
  1}^n \theta_i$. We have $\E[\varphi(x + Z)^4] = 1$ and by assumption the
Fisher information $\E[(f'(Z))^2 / f(Z)^2]$ exists, and the functions $\psi$
and $\chi$ have the desired conditions of
Corollary~\ref{corollary:polyak-juditsky} (as we verify in
Section~\ref{sec:proof-normal-sgd}). Finally, the function $\theta \mapsto
\E_{P_\theta}[(\bar{\theta}_n - \theta)^2]$ is continuous in $\theta$, so
that for $x^\star = \theta$ and $g'' \defeq 1$, \newtext{we may apply
Corollary~\ref{corollary:polyak-mse} to obtain}
\begin{align*}
\ex{(\bar{\theta}_n - \theta)^2} = \frac{1}{4nf(0)^2} + o(n^{-1}).
\end{align*}
\newtext{
From here, existence of the second moment of $\pi$ implies \eqref{eq:adaptive_3}.}



\subsection{Proof of Corollary~\ref{corollary:normal-expansion}}
\label{proof:normal-expansion}

\subsubsection*{Proof of
  Corollary~\ref{corollary:normal-expansion}\eqref{item:regularity}}

The proof of part~\eqref{item:regularity} requires two additional lemmas of
Polyak and Juditsky~\cite{polyak1992acceleration}.
\begin{lem}[\cite{polyak1992acceleration}, Lemma 2]
  \label{lemma:polyak-expansion}
  Define the process $\Delta_i^1 = \Delta_{i-1}^1
  - \gamma_i (A \Delta_{i-1}^1 + \xi_i)$ for $i = 1, 2, \ldots$.
  Assume that $A>0$ and the stepsizes $\gamma_i$ satisfy
  condition~\eqref{eqn:lazy-gamma}. Then
  for $\bar{\Delta}_n^1 = \frac{1}{n} \sum_{i = 1}^n \Delta_i^1$, we have
  \begin{equation}
    \label{eqn:polyak-expansion}
    \sqrt{n} \bar{\Delta}_n^1
    = \frac{\alpha_n \Delta_0^1}{\sqrt{n} \gamma_0}
    + \frac{1}{\sqrt{n} A} \sum_{i=1}^{n-1} \xi_i
    + \frac{1}{\sqrt{n}}\sum_{i=1}^{n-1} w_i^n \xi_i,
  \end{equation}
  where $\alpha_n$ and $w_i^n$ are real numbers such that $|\alpha_n| \leq
  K$ and $|w_i^n|\leq K$ for some $K< \infty$, and $\lim_{n\to \infty}
  \frac{1}{n} \sum_{i=1}^{n-1} |w_i^n| = 0$.
\end{lem} 

\begin{lem}[\cite{polyak1992acceleration}]
  \label{lemma:converging-power-sum}
  Under the conditions of Corollary~\ref{corollary:normal-expansion},
  with probability 1,
  \begin{equation*}
  \sum_{i=1}^\infty \frac{|\Delta_{i}|^{1+\lambda}}{\sqrt{i}} < \infty.
  \end{equation*}
\end{lem}
\noindent
Lemma~\ref{lemma:converging-power-sum} follows from the proof of Theorem 2
in \cite[page 851]{polyak1992acceleration}.

We separate the proof of part~\eqref{item:regularity} into two lemmas, which
mirror the proofs of Polyak and Juditsky~\cite{polyak1992acceleration}; together they
immediately give the result.

\begin{lemma}
  The expansion~\eqref{eq:normal_expansion_lem} holds for the process
  $\bar{\Delta}^1_n$ defined by the iteration
  \begin{align}
    & \Delta_i^1  = \Delta_{i-1}^1 - \gamma_i \psi'(0) \Delta_{i-1}^1 - \gamma_i \varphi(Z_i), \qquad
    \Delta_0^1 = \Delta_0 \nonumber\\
    & \bar{\Delta}^1_n = \frac{1}{n}\sum_{i=0}^{n-1} \Delta^1_i.
    \label{eqn:polyak-expansion_lem1_alg}
  \end{align}
\end{lemma}
\begin{proof}
  To prove this claim, use Lemma~\ref{lemma:polyak-expansion} with $A =
  \psi'(0)$ and $\xi_i = -\varphi(Z_i)$, which by
  condition~\eqref{item:zero-gradient} in
  Corollary~\ref{corollary:polyak-juditsky} gives that
  $\E[\xi_i] = 0$ and that the $\xi_i$ are independent.
  The first term $\alpha_n \Delta_0^1 / \gamma_0 \sqrt{n} \to 0$
  in Eq.~\eqref{eqn:polyak-expansion}. In addition,
  by independence and that the $\xi_i$ are mean-zero, we have
  \begin{align*}
    & \ex{ \left( \frac{1}{\sqrt{n}} \sum_{i=1}^{n-1} w_i^n \xi_i \right)^2 } \\
     & \quad = \frac{1}{n}  \sum_{i=1}^n (w_i^n)^2 \ex{ \xi_i^2} + \frac{1}{n}  \sum_{i\neq j}^n w_i^n w_j^n \ex{ \xi_i \xi_j} \\
    & \quad = \frac{1}{n}  \sum_{i=1}^n (w_i^n)^2 \ex{ \varphi(Z_i)^2} = \chi(0) \frac{1}{n}  \sum_{i=1}^n (w_i^n)^2 \to 0
  \end{align*}
  by Lemma~\ref{lemma:polyak-expansion}.
  Thus, the expansion~\eqref{eqn:polyak-expansion} in
  Lemma~\ref{lemma:polyak-expansion} gives
  \begin{equation*}
    \sqrt{n} \bar{\Delta}^1_n
    = -\frac{1}{\sqrt{n}} \frac{1}{\psi'(0)}
    \sum_{i=1}^{n-1} \varphi(Z_i)+ o_{P,n}(1)
  \end{equation*}
  as desired.
\end{proof}

We then have the following asymptotic equivalence.
\begin{lemma}
  The sequences $\bar{\Delta}_n$ and $\bar{\Delta}^1_n$ are asymptotically
  equivalent, meaning that
  $\sqrt{n} (\bar{\Delta}_n - \bar{\Delta}_n^1) \cp 0$.
\end{lemma}
\begin{proof}
  From the recursions~\eqref{eq:Polyak_Juditsky_alg} and
  \eqref{eqn:polyak-expansion_lem1_alg}, the difference $\delta_i = \Delta_i
  - \Delta_i^1$ satisfies
  \begin{align*}
  & \delta_i = \delta_{i-1} - \gamma_i \psi'(0) \delta_{i-1}  \\
  & \qquad + \gamma_i \left( \psi'(0) \Delta_{i-1}  + \varphi(Z_i) - \varphi(\Delta_{i-1} + Z_i) \right),
  \end{align*}
  where $\delta_0 = 0$. Applying Lemma~\ref{lemma:polyak-expansion} with the
  choices $\xi_i = \psi'(0) \Delta_{i-1} + \varphi(Z_i) -
  \varphi(\Delta_{i-1} + Z_i)$ yields
  \begin{align}
    & \sqrt{n}\bar{\delta}_n
    = \frac{1}{\sqrt{n}} \sum_{i=1}^{n-1}
    \left( \frac{1}{\psi'(0)} + w_i^n \right)  \xi_i  \nonumber \\
    & \quad = \frac{1}{\sqrt{n}} \sum_{i=1}^{n-1}
    \left( \frac{1}{\psi'(0)} + w_i^n \right)
    \left( \psi'(0) \Delta_{i-1}  - \psi(\Delta_{i-1}) \right)
    \label{eq:PJ_proof1} \\
    & \qquad ~ + 
    \frac{1}{\sqrt{n}} \sum_{i=1}^{n-1} \left( \frac{1}{\psi'(0)}
    + w_i^n \right) \label{eq:PJ_proof2} \\
    & \qquad \qquad \times 
    \left( \psi(\Delta_{i-1})  + \varphi(Z_i) - \varphi(\Delta_{i-1}+Z_i)
    \right) \nonumber 
  \end{align}
  For the term \eqref{eq:PJ_proof1},
  the assumption~\eqref{eqn:local-hessian-psi} that
  $|\psi(x) - \psi'(0) x| = O(x^{1 + \lambda})$
  and that $\sup_{i,n} |w_i^n| < \infty$ by Lemma~\ref{lemma:polyak-expansion}
  give that there exists $K < \infty$ such that
  $|\psi'(0)^{-1} + w_i^n| |\psi'(0) \Delta_{i-1} - \psi(\Delta_{i-1})|
  \le K |\Delta_i|^{1 + \lambda}$.
  Lemma~\ref{lemma:converging-power-sum} gives that
  $\sum_{i = 1}^n
  \frac{1}{\sqrt{i}} |\Delta_i|^{1 + \lambda} < \infty$,
  and so the Kronecker lemma gives that
  \begin{equation*}
    \frac{1}{\sqrt{n}} \sum_{i=1}^{n-1} \left( \frac{1}{\psi'(0)}  + w_i^n \right)  \left( \psi'(0) \Delta_{i-1}  - \psi(\Delta_{i-1}) \right) \cas 0.
  \end{equation*}

  The term \eqref{eq:PJ_proof2} is somewhat more challenging to control.
  We define
  \begin{equation*}
    \epsilon_i \defeq \psi(\Delta_{i-1}) + \varphi(Z_i)
    - \varphi(\Delta_{i-1}+Z_i),
  \end{equation*}
  and let
  $\mc{F}_i = \sigma(Z_1, \ldots, Z_i)$ be the $\sigma$-field of the randomness
  through time $i$. We use a square integrable martingale convergence
  theorem~\cite[Exercise~5.3.35]{Dembo16}. Noting that
  $\Delta_i \in \mc{F}_i$, we have
  \begin{align}
    & \E[\epsilon_i^2 \mid \mc{F}_{i-1}] \nonumber \\
    & \quad = \E[(\psi(\Delta_{i-1}) + \varphi(Z_i)
      - \varphi(\Delta_{i-1} + Z_i))^2 \mid \mc{F}_{i-1}]
    \nonumber \\
    & \quad  \le 2 \psi(\Delta_{i-1})^2
    + 2 \E[(\varphi(\Delta_{i-1} + Z_i) - \varphi(Z_i))^2 \mid \mc{F}_{i-1}]
    \nonumber \\
    & \quad \le
    K \left[|\Delta_{i-1}|^{1 + \lambda}
      + |\Delta_{i-1}|^\lambda
      + \Delta_{i-1}^2 \right],
    \label{eqn:bound-psi-error-expectations}
  \end{align}
  where inequality~\eqref{eqn:bound-psi-error-expectations} follows by the
  conditions~\eqref{eqn:local-hessian-psi}
  and~\eqref{eqn:additional-local-smooth}, and $\E[\varepsilon_i \mid
    \mc{F}_{i-1}] = 0$ for all $i$ by definition of $\psi(x) = \E[\varphi(x
    + Z)]$ and that $\psi(0) = 0$.  We now control the expectations of these
  quantities. For $R<\infty$, define the the stopping time $\tau_R \defeq \inf
  \{i : |\Delta_i| > R\}$, which satisfies $\{\tau_R \le i\} \in \mc{F}_i$
  for each $i$. Then using~\cite[Eq.~(A13-A14)]{polyak1992acceleration}, we have
  \begin{equation*}
    \E[\Delta_i^2 \indic{\tau_R > i}] \le K \gamma_i,
  \end{equation*}
  and so inequality~\eqref{eqn:bound-psi-error-expectations} gives that
  \begin{align*}
    & \E\bigg[\sum_{i = 1}^\infty \frac{1}{i}
      |\varepsilon_i|^2 \indic{\tau_R > n} \bigg]
    \le K \sum_{i = 1}^\infty \frac{\gamma_i^\lambda}{i}
    < \infty \\
    & 
    ~~ \mbox{so} ~~
    \sum_{i = 1}^\infty \frac{1}{i}
    \epsilon_i^2 \indic{\tau_R > n} < \infty
    ~ \mbox{a.s.}
  \end{align*}
by Condition~\eqref{eqn:lazy-gamma}. 
%
 As in the proof of Theorems~2 and~4 in \cite{polyak1992acceleration}, the Robbins-Siegmund Theorem \cite{robbins1971convergence} applied to the increment of $|\Delta_t|^2$ implies that for every 
  $\epsilon>0$ there exists some $R'>0$ such that
  \begin{align}
  \label{eq:sup_delta}
  \Prob\left(\sup_i |\Delta_i| \le R' \right) \ge 1-\epsilon.
  \end{align} 
Consequently, there exists some $R'' < \infty$ such that $\tau_{R''} = \infty$. 
We obtain that 
 \begin{equation*}
    \sum_{i = 1}^\infty \frac{1}{i}
    \varepsilon_i^2  < \infty
    ~~ \mbox{a.s.}.
  \end{equation*}
  Applying the square integrable martingale convergence
  theorem of \cite[Ex.~5.3.35]{Dembo16}, we have
  \begin{equation*}
    \frac{1}{\sqrt{n}} \sum_{i = 1}^n
    \left(\frac{1}{\psi'(0)} + w_i^n\right) \epsilon_i \cas 0,
  \end{equation*}
  so that both equations~\eqref{eq:PJ_proof1} and~\eqref{eq:PJ_proof2}
  converge almost surely to 0.
\end{proof}

\subsubsection*{Proof of
  Corollary~\ref{corollary:normal-expansion}\eqref{item:apply-le-cam}}

This is essentially an immediate consequence of Le Cam's third
lemma~\cite[Example 6.7]{VanDerVaart98}.
Recall~\cite[Thm.~7.2]{VanDerVaart98} that if
a family $\{P_\theta\}_{\theta \in \Theta}$ is quadratic mean differentiable
at $\theta$ with score $\score_\theta$, then
it is LAN at $\theta$ (Definition~\ref{definition:lan})
with information matrix $I_\theta = \E[\score_\theta \score_\theta^\top]$.




The regularity result~\eqref{eq:normal_expansion_lem} gives
\begin{equation*}
  \sqrt{n} \bar{\Delta}_n = -\frac{1}{\sqrt{n}}
  \sum_{i = 1}^n  \frac{\varphi(Z_i)}{\psi'(0)} + o_{P,n}(1).
\end{equation*}
The conditions in
Corollary~\ref{corollary:normal-expansion}\eqref{item:apply-le-cam} imply
that the Fisher information $I_h =
\E_h[\score_h(Z_1)^2]$ exists and is continuous for
$\score_h(z) = \frac{p'(z - h)}{p(z - h)}$,
and the asymptotic expansion Definitions~\ref{definition:qmd}
and~\ref{definition:lan} combined
with the preceding display, give the joint convergence
\begin{align*}
  & \left(\sqrt{n} \bar{\Delta}_n, \sum_{i=1}^n \log \frac{P_{h_n/\sqrt{n}}}{P_0}(Z_i)
  \right) \cd \normal \left(\mu, \Sigma \right), \\
\end{align*}  
where
\begin{align*}
  \mu & = \left(0,-\frac{h^2}{2} I_0 \right), ~~~\mbox{and} \\
  \Sigma & = \begin{pmatrix}
    \frac{\chi(0)}{\psi'(0)^2} & \frac{-h}{ \psi'(0)}
    \E_p[\varphi(Z_1) \score_0(Z_1)] \\
    \frac{-h}{ \psi'(0)} \E_p[\varphi(Z_1) \score_0(Z_1)]
    & h^2 I_0
  \end{pmatrix}.
\end{align*}
Le Cam's third lemma \cite[Exm. 6.7]{VanDerVaart98} then implies the
convergence
\begin{equation*}
  \sqrt{n}\bar{\Delta}_n
  \mathop{\cd}_{P_{h_n/\sqrt{n}}^n}
  \Ncal\left(\frac{-h}{ \psi'(0)} \E_p[\varphi(Z) \score(Z)],
  \frac{\chi(0)}{\psi'(0)^2} \right)
\end{equation*}
under the alternatives $P^n_{h_n/\sqrt{n}}$, which gives
Corollary~\ref{corollary:normal-expansion}\eqref{item:apply-le-cam}.

\subsection{Proof of Lemma~\ref{lemma:qmd}}
\label{sec:proof-qmd}

The proof is essentially completely parallel to that of \cite[Lemma
  7.6]{VanDerVaart98}. Define $\dot{s}_\theta = \half
\frac{\dot{p}_\theta}{p_\theta} \sqrt{p_\theta}$, which exists $\mu$-almost
surely, so that $\int \dot{s}_\theta \dot{s}_\theta^\top d\mu$ is
well-defined (though it may be infinite). By Lebesgue's integration theorem,
we have
\begin{equation*}
  s_{\theta + h}(x) - s_\theta(x) = \int_0^1 h^\top \dot{s}_{\theta + t h}(x) dt,
\end{equation*}
and so By Jensen's inequality (or Cauchy-Schwartz) we have
\begin{equation*}
  (s_{\theta + h}(x) - s_\theta(x))^2
  \le \int_0^1 h^\top \dot{s}_{\theta + t h}(x)
  \dot{s}_{\theta + t h}(x) ^\top h dt.
\end{equation*}
Thus, for any $h_t$ we have
\begin{align*}
  & \int \left(\frac{s_{\theta + t h_t}(x) - s_\theta(x)}{t}\right)^2
  d\mu(x) \\
  & \qquad \le \int \int_0^1 (h_t^\top \dot{s}_{\theta + u t h_t})^2 du
  d\mu \\
  & \qquad  = \int_0^1 h_t^\top \int
  \dot{s}_{\theta + u t h_t}\dot{s}_{\theta + u t h_t}^\top
  d\mu(x) h_t  du \\
  & \qquad = \frac{1}{4} h_t^\top \left(\int_0^1 I_{\theta + u t h_t} du\right) h_t.
\end{align*}
By continuity, as $h_t \to h$ and $t \to 0$ the assumed continuity
of $\theta \mapsto I_\theta$ gives that the final display converges
to $h^\top I_\theta h$.

Now, we note that
\begin{equation*}
  \lim_{t \downarrow 0}
  \left(\frac{s_{\theta + t h_t}(x) - s_\theta(x)}{t}
  - h^\top \dot{s}_\theta(x)\right)^2 = 0
\end{equation*}
for all $x$ excepting a $\mu$-null set, and the
variant of the dominated convergence theorem in
\cite[Prop.~2.29]{VanDerVaart98} implies that
\begin{align*}
  & \lim_{t \to 0} 
  \frac{1}{t^2}
  \int \left(s_{\theta + t h_t}(x) - s_\theta(x)
  - t h^\top \dot{s}_\theta(x)\right)^2 d\mu(x) \\
  & ~~ = \lim_{t \to 0}
  \int \left(\frac{s_{\theta + t h_t}(x) - s_\theta(x)}{t}
  - h^\top \dot{s}_\theta(x)\right)^2 d\mu(x)
  = 0,
\end{align*}
completing the proof.


\section{Proof of Theorem~\ref{theorem:non-adaptive-minimax}}
\label{sec:proof-non-adaptive-minimax}

We follow a similar outline to the optimality results we establish
in the proof of Theorem~\ref{thm:sgd}\eqref{item:sgd-regular} in
Sec.~\ref{sec:proof-sgd-regular}.
Roughly, we establish that the family $P_\theta$ of distributions
on the bits $B_i$ is locally asymptotically normal
(Definition~\ref{definition:lan}) via a quadratic
mean differentiability argument. After this, the result
follows by standard local asymptotic minimax theory.

We begin with an argument on the smoothness properties of the densities,
which is important for our Taylor expansions to come.
\begin{lemma}
  \label{lemma:derivative-bounds}
  Let Assumption~\ref{assumption:detection-regions}\eqref{item:lipschitz-density} hold. Then for any $A = \cup_{i = 1}^k
  [a_i, b_i]$ and $h \in \R$,
  \begin{align}
    \left|P_{\theta + h}(A) - P_\theta(A) - \dPtheta(A) h \right| \leq 
    k \cdot \lip(f) h^2,
    \label{eqn:expansion-dPtheta}
  \end{align}
  where
  \begin{align*}
    \dPtheta(A) = \sum_{i = 1}^k f(a_i - \theta) - f(b_i - \theta).
  \end{align*}
  Additionally, we have the bounds
  \begin{align}
    & \label{eqn:bound-density-diffs}
    |f(b) - f(a)| \le 2 \sqrt{\lip(f) P([a, b])} \\
    & 
    ~~ \mbox{and} ~~
    |\dPtheta(A)| \le 2 \sqrt{k \lip(f)}. \nonumber
  \end{align}
\end{lemma}
\noindent
See Section~\ref{sec:proof-derivative-bounds} for a proof.

The second lemma provides the local asymptotic normality we require.
\begin{lemma}
  \label{lemma:lan-bits}
  Let
  Assumption~\ref{assumption:detection-regions}\eqref{item:lipschitz-density}
  and~\eqref{item:finite-intervals} hold, and let $B_i = \indic{X_i \in
    A_i}$.  Let $h_n \to h \in \R$. Then for any $\theta \in
  \mbox{int}\Theta$,
  \begin{align*}
    & \sum_{i = 1}^n \log \frac{P_{\theta + h_n/\sqrt{n}}(B_i)}{
      P_\theta(B_i)} \\
    & \qquad = \frac{h}{\sqrt{n}}
    \sum_{i = 1}^n \score_\theta(B_i) - \frac{h^2}{4n} \sum_{i = 1}^n \var(\score_\theta(B_i)) \\
    & \qquad \qquad - \frac{h^2}{4n} \sum_{i = 1}^n \score_\theta(B_i)^2
    + o_P(1).
  \end{align*}
  If additionally
  Assumption~\ref{assumption:detection-regions}\eqref{item:limit-variance}
  holds, then
  \begin{equation*}
    \sum_{i = 1}^n \log \frac{P_{\theta + h_n/\sqrt{n}}(B_i)}{
      P_\theta(B_i)}
    = \frac{h}{\sqrt{n}}
    \sum_{i = 1}^n \score_\theta(B_i)
    - \frac{h^2}{2} \kappa(\theta) + o_P(1).
  \end{equation*}
\end{lemma}
\noindent
The proof of Lemma~\ref{lemma:lan-bits} is quite technical,
so we defer it to Section~\ref{sec:proof-lan-bits}.

With this lemma, it is not too challenging to demonstrate the local
asymptotic normality (Definition~\ref{definition:lan}) of the family
$\{P_\theta\}$. Indeed, Lemma~\ref{lemma:derivative-bounds} guarantees that
$|\dPtheta(A_n)| \le 2\sqrt{k_n \lip(f)}$ for all $n$, so that
$\E_\theta[|\score_\theta(B_i)|^3] \le C \frac{k_i^{3/2}
  \lip(f)^{3/2}}{P_\theta(A_i)^2 (1 - P_\theta(A_i))^2}$, while
Assumption~\ref{assumption:detection-regions}\eqref{item:finite-intervals}
guarantees that
$\frac{1}{n^3} \sum_{i = 1}^n \E_\theta[|\score_\theta(B_i)|^3]
\to 0$. Because $\E[\score_\theta(B_i)] = 0$,
the Lyapunov central limit theorem
applies to give
\begin{equation*}
  \frac{1}{\sqrt{n}} \sum_{i = 1}^n \score_\theta(B_i)
  \cd \normal\left(0, \kappa(\theta)\right)
\end{equation*}
under
Assumption~\ref{assumption:detection-regions}\eqref{item:limit-variance},
so that the family $\{P_\theta\}$ is locally asymptotically normal
(Def.~\ref{definition:lan}).

We now recall the familiar H\'{a}jek-Le-Cam local asymptotic minimax
result~\cite[Thm.~8.11]{VanDerVaart98}: if the family
$\{P_\theta\}$ is LAN with precision $\kappa(\theta)$, then
\begin{align*}
  \liminf_{c \to \infty} & \liminf_n \sup_{\norm{\tau - \theta} \le
    c / \sqrt{n}} \E_\tau\left[L(\sqrt{n}(\theta_n - \tau))\right] \\
  & \qquad  \ge \E[L(Z / \sqrt{\kappa(\theta)})]
\end{align*}
for any symmetric quasi-convex loss $L$, where $Z \sim \normal(0, 1)$.
This immediately gives Theorem~\ref{theorem:non-adaptive-minimax}.





\subsection{Proof of Lemma~\ref{lemma:derivative-bounds}}
\label{sec:proof-derivative-bounds}

To see the first claim of the lemma, we consider the simpler special case
that $A = [a, b]$. Then as $f$ is Lipschitz (and hence absolutely
continuous and a.e.\ differentiable with $\linf{f'} \le \lip(f)$),
we have
\begin{align*}
  & P_{\theta + h}(A)
  - P_\theta(A)
   = \int_a^b (f(z - \theta - h) - f(z - \theta)) dz \\
  & \qquad = -\int_a^b \int_0^h f'(z - \theta - u) du dz \\
  & \qquad = -\int_0^h \int_a^b f'(z - \theta - u) dz du \\
  & \qquad = \int_0^h f(a - \theta - u) - f(b - \theta - u) du \\
  & \qquad \lesseqgtr \int_0^h (f(a - \theta) - f(b - \theta)) du
  \pm 2 \int_0^h \lip(f) u du \\
  & \qquad = \left[f(a - \theta) - f(b - \theta)\right] h \pm \lip(f) h^2.
\end{align*}
This gives the first two claims of the lemma.

For the second, we require a bit more work.
Let $L = \lip(f)$ for shorthand. Let $a < b$.
Then we always have
\begin{align}
  \label{eqn:interval-length-thing}
  & P([a, b]) \ge \int_a^b f(z) dz \\
  & \quad \ge \int_a^b \max\{f(b) - L (b - z), f(a) - L(z - a), 0\} dz. \nonumber
\end{align}
If $f(a) + f(b) \ge L (b - a)$, then the point $\hat{z} =
\frac{a + b}{2} - \frac{f(b) - f(a)}{2L}$ satisfies both $f(b) - L(b -
\hat{z}) \ge 0$ and $f(a) - L(\hat{z} - a) \ge 0$. The
integral~\eqref{eqn:interval-length-thing} then becomes
\begin{align*}
  \lefteqn{\int_a^{\hat{z}}
    (f(a) - L (z - a) dz)
    + \int_{\hat{z}}^b
    (f(b) - L (b - z)) dz} \\
  & = \medmath{\frac{f(a) + f(b)}{2}
  \left(\frac{b - a}{2}\right)
  - L \left(\frac{b - a}{2} \right)^2
  + \frac{(f(b) - f(a))^2}{4L}},
\end{align*}
and using the assumption that $\frac{f(a) + f(b)}{2} \ge L(b - a)$,
we obtain
\begin{align*}
  \frac{(f(b) - f(a))^2}{4L}
  & \le \frac{f(b) + f(b)}{2} \frac{b - a}{2} \\
  & \qquad - L \left(\frac{b - a}{2}\right)^2
  + \frac{(f(b) - f(a))^2}{4L} \\
  & \le P([a, b]).
\end{align*}
That is, $|f(b) - f(a)| \le 2 \sqrt{\lip(f) P([a, b])}$.
In the converse case that $f(a) + f(b) \le L(b - a)$, then
the integral~\eqref{eqn:interval-length-thing} becomes
\begin{align*}
  P([a, b])
  & \ge \int_a^{a + \frac{f(a)}{L}}
  (f(a) - L (z - a)) dz \\
  & \qquad + \int_{b - \frac{f(b)}{L}}^b (f(b) - L (b - z)) dz \\
  & = \frac{f(a)^2}{L}
  - \frac{f(a)^2}{2L}
  + \frac{f(b)^2}{L}
  - \frac{f(b)^2}{2L},
\end{align*}
so that 
\begin{align*}
\newtext{\frac{f(a) + f(b)}{\sqrt{2}} \le} \sqrt{f(a)^2 + f(b)^2} \le \sqrt{2 \lip(f) P([a, b])},
\end{align*}
\newtext{where the left inequality follows from concavity of $\sqrt{\cdot}$.} In sum, we have demonstrated that always the first
bound~\eqref{eqn:bound-density-diffs} holds.
To show the second inequality in expression~\eqref{eqn:bound-density-diffs},
note that
$\sum_i P([a_i, b_i]) \le 1$, and apply Cauchy-Schwarz.

\subsection{Proof of Lemma~\ref{lemma:lan-bits}}
\label{sec:proof-lan-bits}

Our proof follows that of \cite[Thm.~7.2]{VanDerVaart98} closely. We first
demonstrate a type of uniform quadratic mean differentiability
(Definition~\ref{definition:qmd}) for sets $A$ that are finite unions of
intervals. By a Taylor approximation and concavity of
$\sqrt{\cdot}$, we have
\begin{equation*}
  \sqrt{a} + \frac{b}{2 \sqrt{a}} -
  \frac{b^2}{4a^{\newtext{3/2}}}
  \le \sqrt{a + b} \le \sqrt{a}
  + \frac{b}{2 \sqrt{a}}
\end{equation*}
for any $a > 0$ and $|b| \le 3a/4$. Consequently,
recalling that $\score_\theta(A) = \dPtheta(A) / P_\theta(A)$,
for any $h \in \R$ and $A = \cup_{i = 1}^k [t_i^-, t_i^+]$ the
union of $k$ intervals, the expansion~\eqref{eqn:expansion-dPtheta} yields
\begin{align*}
  & \left(\sqrt{P_{\theta + h}(A)} -
  \sqrt{P_\theta(A)} - \half h \score_\theta(A) \sqrt{P_\theta(A)}
  \right)^2 \\
  & \quad \le
  \left(
  \frac{k \lip(f)}{2 \sqrt{P_\theta(A)}} h^2
  + \frac{(|\dPtheta(A) h| + h^2 \lip(f))^2}{P_\theta(A)^{3/2}}
  \right)^2,
\end{align*}
valid for $h$ such that
$|\dPtheta(A) h| \le P_\theta(A) / 4$ and
$k h^2 \lip^2(f) \le P_\theta(A) / 4$.
Thus, under
Assumption~\ref{assumption:detection-regions}\eqref{item:finite-intervals},
there exists a numerical constant $C < \infty$ such that
\begin{subequations}
  \label{eqn:h-fourth}
  \begin{align}
    & \nonumber \left(\sqrt{P_{\theta + h}(A)} -
    \sqrt{P_\theta(A)} - \half
    h \score_\theta(A) \sqrt{P_\theta(A)}\right)^2 \\
    & \nonumber \quad \le
    \left(\frac{h^2 k \cdot \lip(f)}{2 \sqrt{P_\theta(A)}}
    + \frac{(|\dPtheta(A) h| + k h^2 \lip(f))^2}{
      P_\theta(A)^{3/2}}\right)^2 \\
    & \quad \le \frac{C}{P_\theta(A)} \left[k^2 \lip^2(f)
      + \score_\theta(A)^2
      + \frac{k^4 h^4 \lip^4(f)}{P_\theta(A)^2}
      \right] \cdot h^4,
  \end{align}
  valid whenever $|\dPtheta(A) h|
  \le P_\theta(A) / 4$ and $k h^2 \lip^2(f) \le P_\theta(A) / 4$,
  and similarly, we have
  \begin{align}
    & \left(\sqrt{P_{\theta + h}(A^c)} -
    \sqrt{P_\theta(A^c)} - \half
    h \score_\theta(A^c) \sqrt{P_\theta(A^c)}\right)^2 \nonumber \\
    & ~~ \le \frac{C}{P_\theta(A^c)}
    \left[k^2 \lip^2(f)
      + \score_\theta(A^c)^2
      + \frac{k^4 h^4 \lip^4(f)}{P_\theta(A^c)^2}
      \right] \cdot h^4. 
  \end{align}
\end{subequations}
That is, the family $\{P_\theta\}$ with bit observations $B_n$ satisfies a
uniform type of quadratic-mean differentiability
(Def.~\ref{definition:qmd}).

For shorthand, define $P_n = P_{\theta + h_n / \sqrt{n}}$ and $P =
P_\theta$, and let $p_n, p$ be shorthand for the p.m.f.s of the two
distributions.  For the sets $A_i$ we recall that $B_i = \indic{X_i \in
  A_i}$.  The random variables
\begin{equation*}
  W_{n,i} \defeq 2 \left(\sqrt{\frac{p_n}{p}}(B_i) - 1\right)
\end{equation*}
are with $P$-probability 1 well-defined, and by the
inequalities~\eqref{eqn:h-fourth}, we have
that
\begin{align}
\label{eqn:var-wni-expansion}
  & \var\left(W_{n,i} - \frac{h_n}{\sqrt{n}} \score_\theta(B_i)\right) \\
    & \quad \le C \frac{k_i^2 \lip^2(f) + \score_\theta(A_i)^2
      + \score_\theta(A_i^c)^2}{P_\theta(A_i) P_\theta(A_i^c)}
    \cdot \frac{h_n^4}{n^2} \nonumber \\
    & \qquad + C \frac{k^4 \lip^4(f)}{
      P_\theta(A_i)^3 P_\theta(A_i^c)^3}
    \frac{h_n^8}{n^4} \nonumber \\
  & \quad \le 
  C \frac{k_i^2 \lip^2(f) + \score_\theta(A_i)^2
  + \score_\theta(A_i^c)^2}{P_\theta(A_i) P_\theta(A_i^c)}
  \cdot \frac{h_n^4}{n^2} \\
  & \qquad + C \frac{k^4 \lip^4(f)}{
    P_\theta(A_i)^3 P_\theta(A_i^c)^3}
  \frac{h_n^8}{n^4}
\end{align}
whenever
\begin{align*}
  & \frac{h}{\sqrt{n}} \max\{\score_\theta(A_i),
  \score_\theta(A_i^c)\}
  \le \frac{1}{4} \\
  & ~~ \mbox{and} ~~
  \frac{k_i h_n^2}{n} \lip^2(f)
  \le \frac{\min\{P_\theta(A_i), P_\theta(A_i^c)\}}{4}
\end{align*}
Now, we use
Assumption~\ref{assumption:detection-regions}\eqref{item:finite-intervals},
coupled with Lemma~\ref{lemma:derivative-bounds} to show that the summed
variances converge to zero.  Indeed, Lemma~\ref{lemma:derivative-bounds} and
inequality~\eqref{eqn:var-wni-expansion} give that
\begin{align*}
  & \var\left(W_{n,i} - \frac{h_n}{\sqrt{n}} \score_\theta(B_i)\right)
   \le C \cdot
  \left[\frac{k_i^2}{P_\theta(A_i) P_\theta(A_i^c)}
    \frac{1}{n} \right. \\
    & \left. \qquad + \frac{k_i}{P_\theta(A_i) P_\theta(A_i^c)}
    \frac{1}{n}
    + \frac{k_i^4}{P_\theta(A_i)^3 P_\theta(A_i^c)^3}
    \frac{1}{n^3}\right] \frac{1}{n},
\end{align*}
where $C < \infty$ depends only on $\lip(f)$ and $h_n$ (both of which are
uniformly bounded) whenever
\begin{equation*}
  \frac{k_i}{P_\theta(A_i) P_\theta(A_i^c)} \frac{1}{n} \le
  \frac{1}{C}.
\end{equation*}
Assumption~\ref{assumption:detection-regions}\eqref{item:finite-intervals}
thus implies that $\E[\score_\theta(B_i)] = 0$ and
\begin{align}
  & \var\left(\sum_{i = 1}^n W_{n,i} - \frac{h_n}{\sqrt{n}} \score_\theta(B_i)
  \right) \nonumber \\
  & \qquad = \sum_{i = 1}^n \var\left(W_{n,i} - \frac{h_n}{\sqrt{n}} \score_\theta(B_i)
  \right)
  \to 0. \label{eqn:summed-variances-to-zero}
\end{align}

We now control the expectation of the $W_{n,i}$. Defining $\mu_i$ to be the
induced counting measure on $B_i = \indic{X_i \in A_i}$,
\begin{align*}
  & \sum_{i = 1}^n \E[W_{n,i}]
  = 2\sum_{i = 1}^n
  \left(\int \sqrt{p_n(b)} \sqrt{p(b)} d\mu_i(b) - 1 \right) \\
  & \quad = -\sum_{i = 1}^n \int \left(\sqrt{p_n(b)} - \sqrt{p(b)}\right)^2
  d\mu_i(b) \\
  & \quad = -\frac{h_n^2}{4 n} \sum_{i = 1}^n \E[\score_\theta(B_i)^2] \\
  & \qquad 
  - \sum_{i = 1}^n \int \left(\sqrt{p_n(b)} - \sqrt{p(b)}
  - \frac{h_n}{\sqrt{n}} \score_\theta(b) \sqrt{p(b)}\right)^2 d\mu_i(b) \\
  & \qquad \medmath{
  - \sum_{i = 1}^n \int \left(\sqrt{p_n(b)} - \sqrt{p(b)}
  - \frac{h_n}{\sqrt{n}} \score_\theta(b) \sqrt{p(b)}\right)
  \frac{h_n}{\sqrt{n}} \score_\theta(b) \sqrt{p(b)}d\mu_i(b)} \\
  & = -\bigg(\frac{h^2}{4n} \sum_{i = 1}^n \E[\score_\theta(B_i)^2]\bigg)
  - o(1)
\end{align*}
uniformly in $h$, with a derivation completely paralleling that above.
Therefore, we obtain
\begin{align*}
   \sum_{i = 1}^n W_{n,i} & = \sum_{i = 1}^n \left(W_{n,i} - \frac{h_n}{\sqrt{n}} \score_\theta(B_i)\right)
  + \frac{h_n}{\sqrt{n}} \sum_{i = 1}^n \score_\theta(B_i) \\
  & = -\frac{h^2}{4n} \sum_{i = 1}^n \E[\score_\theta(B_i)^2]
  + \frac{h}{\sqrt{n}} \sum_{i = 1}^n \score_\theta(B_i)
  + o_P(1),
\end{align*}
where we have used that $h_n \to h$.

Now, we write the log-likelihood ratio. We have
\begin{align*}
  & \sum_{i = 1}^n \log \frac{p_n(B_i)}{p(B_i)}
  = 2 \sum_{i = 1}^n \log\left(1 + \half W_{n,i}\right) \\
  & \qquad  = \sum_{i = 1}^n W_{n,i}
  - \frac{1}{4} \sum_{i = 1}^n W_{n,i}^2
  + \half \sum_{i = 1}^n W_{n,i}^2 R(W_{n,i})
\end{align*}
where the remainder $|R(W_{n,i})| \le |W_{n,i}|$ for $|W_{n,i}| \le 1$.
Using the Taylor expansions of $\sqrt{\cdot}$ and
Lemma~\ref{lemma:derivative-bounds}, we have
\begin{align}
  &\left| \half W_{n,i} \right|
   \leq \half \left| \score_\theta(B_i) \right|
  \frac{h_n}{\sqrt{n}} \nonumber \\
  & \qquad + \left|\frac{h_n^2}{n} \frac{k_i \lip(f)}{p(B_i)}
  + \frac{h_n^2}{n} \score_\theta(B_i)^2
  + \frac{ h_n^4}{n^2} \frac{k_i^2 \lip(f)^2}{p(B_i)^2}\right| \nonumber \\
  & \quad = \half \score_\theta(B_i)
  \frac{h_n}{\sqrt{n}} \nonumber \\
  & \qquad + 
  C \left|\frac{\sqrt{k_i}}{\sqrt{n} p(B_i)}
  + \frac{k_i}{n p(B_i)}
  + \frac{\sqrt{k_i}}{p(B_i)^2 n}
  + \frac{k_i^2}{p(B_i)^2 n^2}
  \right|
  \label{eqn:max-wni-zero}
\end{align}
where $|C| < \infty$ depends only on $\lip(f)$ and $h_n$ and so is
uniformly bounded. From Assumption~\ref{assumption:detection-regions}\eqref{item:finite-intervals} we get
\[
C \left|\frac{\sqrt{k_i}}{\sqrt{n} p(B_i)}
  + \frac{k_i}{n p(B_i)}
  + \frac{\sqrt{k_i}}{p(B_i)^2 n}
  + \frac{k_i^2}{p(B_i)^2 n^2}
  \right| \to 0. 
\]
Consequently $\max_i W_{n,i} \to 0$, so that
\begin{align}
  \label{eqn:almost-at-the-end}
  \sum_{i = 1}^n \log \frac{p_n(B_i)}{p(B_i)}
  & = \frac{h_n}{\sqrt{n}} \sum_{i = 1}^n \score_\theta(B_i)
  - \frac{1}{4} \sum_{i = 1}^n \E[\score_\theta(B_i)^2] \\
  & \qquad - \frac{1}{4} \sum_{i = 1}^n W_{n,i}^2 + o_P(1). \nonumber
\end{align}

It remains to compute $\E[W_{n,i}^2]$. Using the bounds
that $|\score_\theta(B_i)| \le C \sqrt{k_i} / p(B_i)$ from
Lemma~\ref{lemma:derivative-bounds}, the 
expansion~\eqref{eqn:max-wni-zero} yields
\begin{align*}
  \lefteqn{\left|\E\left[W_{n,i}^2 - \frac{h_n^2}{2n} \score_\theta(B_i)^2\right]
    \right|} \\
  &
  \le \frac{C}{n} \left[
    \frac{k_i^{3/2}}{p(A_i)(1 - p(A_i)) \sqrt{n}}
    + \frac{k_i^{3/2}}{p(A_i)^2 (1 - p(A_i))^2 \sqrt{n}} \right. \\
    & \left. ~~ \qquad + \frac{k_i^2}{p(A_i)(1 - p(A_i))} \frac{1}{n^{3/2}}
    \right] \\
  & \quad \quad +
  \frac{C}{n}
  \left[\frac{k_i^2}{p(A_i)(1 - p(A_i))} \frac{1}{n}
    + \frac{k_i}{p(A_i)^3(1 - p(A_i))^3} \frac{1}{n} \right. \\
    & \left. ~~ \qquad 
    + \frac{k_i^4}{p(A_i)^3 (1 - p(A_i))^3} \frac{1}{n^3} \right],
\end{align*}
where $C$ depends only on $h$ and $\lip(f)$.
Thus
\begin{equation*}
  \sum_{i = 1}^n W_{n,i}^2
  = \frac{h_n^2}{n} \sum_{i = 1}^n \score_\theta(B_i)^2
  + o(1),
\end{equation*}
giving Lemma~\ref{lemma:lan-bits}.


\section{Proof of Theorem~\ref{thm:non_existence}}
\label{proof:thm:non_existence}

Let $\Xi$ be the set of points $\theta \in \Theta$ for which $\kappa(\theta)
= \eta(0)$.
Since $B_1,B_2,\ldots$ satisfy the conditions in
Theorem~\ref{theorem:non-adaptive-minimax}, $\theta$ is in $\Xi$ if and only
if $\lim_{n\to \infty} L_n(A_1,\ldots,A_n;\theta) = \eta(0)$. By
assumption, we have $B_i = \indic{X_i \in A_i}$, $A_i =
\cup_{k=1}^K (t^-_{i,k},t^+_{i,k})$, where $t^-_{i,1} \leq t^+_{i,1} \leq \ldots
\leq t^-_{i,K} \leq t^+_{i,K}$, and $t^-_{i,1}$ and $t^+_{i,K}$ may take the values
$-\infty$ and $\infty$, respectively. Denote the set of endpoints
\begin{equation*}
  E_i = \bigcup_{k=1}^{K}\{t^-_{i,k},t^+_{i,k}\},
\end{equation*}
and for $\theta$ and $\epsilon>0$, define
\begin{equation*}
  S_n(\theta, \epsilon) \triangleq \left\{ i\leq n ~ \mbox{s.t.}~
  (\theta-\epsilon,\theta+\epsilon) \cap E_i \neq \emptyset \right\}
\end{equation*}
In words, $S_n$ contains all integers smaller than $n$ in which an
$\epsilon$-ball around $\theta$ contains an endpoint of one of the intervals
defining $A_i$.
We now claim that 
if $\theta \in \Xi$ then $\card(S_n(\theta, \epsilon))/n \to 1$. 
Indeed, for such $\theta$ we have
\begin{align}
& L_n(A_1,\ldots,A_n; \theta) \nonumber \\
& = \frac{1}{n} \sum_{i \in S_n(\epsilon,\theta)}  
\frac{ \left(\sum_{k=1}^{K}  f(\theta - t^+_{i,k})- f(\theta - t^-_{i,k}) \right)^2}{ \sum_{k=1}^{K} \left( F(\theta - t^+_{i,k})- F(\theta - t^-_{i,k}) \right)} \nonumber \\
& \qquad \times \frac{1} {\left(1-\sum_{k=1}^{K} \left( F(\theta - t^+_{i,k})- F(\theta - t^-_{i,k}) \right)\right)} \nonumber \\
& 
+ \frac{1}{n}\sum_{i \notin S_n(\epsilon,\theta) } \frac{ \left(\sum_{k=1}^{K}  f(t^+_{i,k}-\theta) - f(t^-_{i,k}-\theta) \right)^2} { \sum_{k=1}^{K} \left( F(\theta - t^+_{i,k})- F(\theta - t^-_{i,k}) \right)} \nonumber \\
& \qquad \times \frac{1} {\left(1-\sum_{k=1}^{K} \left( F(\theta - t^+_{i,k})- F(\theta - t^-_{i,k}) \right)\right)} \nonumber \\
& \le
\frac{\card\left(S_n(\theta,\epsilon)\right)}{n} \eta(0) + \frac{n-\card\left(S_n(\theta,\epsilon) \right) }{n} \eta(\epsilon) 
 \label{eq:non_existence_proof1}
\end{align}
where the last transition follows from Lemma~\ref{lem:bound_intervals_delta} with $\delta =0$ and the fact that for $i \in S_n(\theta, \epsilon)$, 
\begin{equation*}
\max\left\{ \max_k \eta(t^+_{i,k}-\theta) , \max_k \eta(t^-_{i,k}-\theta)  \right\} \leq \eta(\epsilon) < \eta(0). 
\end{equation*}
Unless  $\card \left(S_n(\theta, \epsilon) \right)/n \to 1$, we get that \eqref{eq:non_existence_proof1}, hence $L_n(A_1,\ldots,A_n ; \theta)$, are bounded from above by a constant that is smaller then $\eta(0)$ in contradiction to the fact that $\theta \in \Xi$.

Assume for the sake of contradiction that there exists $N \geq 2K + 1$
distinct elements $\theta_1,\ldots,\theta_N \in \Xi$. Since each $A_i$
consists of at most $K$ intervals, we have that
\begin{equation}
  \label{eq:few_optimality_points_proof}
  \card \left(\bigcup_{i=1}^n  A_i \right) \leq 2 n K. 
\end{equation}
Fix $\epsilon>0$ such that 
\begin{equation*}
  \epsilon < \frac{1}{2}\min_{i\neq j} |\theta_i - \theta_j|. 
\end{equation*}
Since for each $\theta \in \Theta$ we have $\card \left(S_n(\theta, \epsilon) \right) \to 1$, there exists $n$ large enough such that 
\begin{equation*}
  \card \left(S_n(\theta_i, \epsilon) \right)
  \geq n \left(1-\frac{1}{2N} \right)
\end{equation*}
for all $i=1,\ldots,N$. However, $S_n(\theta_1,\epsilon), \ldots
S_n(\theta_N,\epsilon)$ are disjoint, so the cardinality of their union is
at least $n\left(1-\frac{1}{2N} \right)N > 2nK + n/2$, a contradiction to
inequality~\eqref{eq:few_optimality_points_proof}.

\bibliographystyle{IEEEtran}
\bibliography{IEEEfull,onebit,bib}

\begin{thebibliography}{10}
\providecommand{\url}[1]{#1}
\csname url@samestyle\endcsname
\providecommand{\newblock}{\relax}
\providecommand{\bibinfo}[2]{#2}
\providecommand{\BIBentrySTDinterwordspacing}{\spaceskip=0pt\relax}
\providecommand{\BIBentryALTinterwordstretchfactor}{4}
\providecommand{\BIBentryALTinterwordspacing}{\spaceskip=\fontdimen2\font plus
\BIBentryALTinterwordstretchfactor\fontdimen3\font minus
  \fontdimen4\font\relax}
\providecommand{\BIBforeignlanguage}[2]{{%
\expandafter\ifx\csname l@#1\endcsname\relax
\typeout{** WARNING: IEEEtran.bst: No hyphenation pattern has been}%
\typeout{** loaded for the language `#1'. Using the pattern for}%
\typeout{** the default language instead.}%
\else
\language=\csname l@#1\endcsname
\fi
#2}}
\providecommand{\BIBdecl}{\relax}
\BIBdecl

\bibitem{KipnisAllerton2017}
A.~Kipnis and J.~C. Duchi, ``Mean estimation from adaptive one-bit
  measurements,'' in \emph{55th Annual Allerton Conference on Communication,
  Control, and Computing (Allerton)}, Oct 2017, pp. 1000--1007.

\bibitem{LesserOrTa03}
V.~Lesser, C.~Ortiz, and M.~Tambe, Eds., \emph{{Distributed Sensor Networks: A
  Multiagent Perspective}}.\hskip 1em plus 0.5em minus 0.4em\relax Kluwer
  Academic Publishers, 2003, vol.~9.

\bibitem{LiWoHuSa02}
D.~Li, K.~Wong, Y.~Hu, and A.~Sayeed, ``Detection, classification and tracking
  of targets in distributed sensor networks,'' in \emph{IEEE Signal Processing
  Magazine}, 2002, pp. 17--29.

\bibitem{FullerMi11}
S.~Fuller and L.~Millett, \emph{The Future of Computing Performance: Game Over
  or Next Level?}\hskip 1em plus 0.5em minus 0.4em\relax National Academies
  Press, 2011.

\bibitem{1092194}
J.~Candy, ``A use of limit cycle oscillations to obtain robust
  analog-to-digital converters,'' \emph{{IEEE} Transactions on Communications},
  vol.~22, no.~3, pp. 298--305, Mar 1974.

\bibitem{53738}
P.~W. Wong and R.~M. Gray, ``Sigma-delta modulation with i.i.d. {G}aussian
  inputs,'' \emph{{IEEE} Transactions on Information Theory}, vol.~36, no.~4,
  pp. 784--798, Jul 1990.

\bibitem{DuchiJoWa18}
J.~C. Duchi, M.~I. Jordan, and M.~J. Wainwright, ``Minimax optimal procedures
  for locally private estimation (with discussion),'' \emph{Journal of the
  American Statistical Association}, vol. 113, no. 521, pp. 182--215, 2018.

\bibitem{baraniuk2017exponential}
R.~G. Baraniuk, S.~Foucart, D.~Needell, Y.~Plan, and M.~Wootters, ``Exponential
  decay of reconstruction error from binary measurements of sparse signals,''
  \emph{{IEEE} Transactions on Information Theory}, vol.~63, no.~6, pp.
  3368--3385, 2017.

\bibitem{jacques2013robust}
L.~Jacques, J.~N. Laska, P.~T. Boufounos, and R.~G. Baraniuk, ``Robust 1-bit
  compressive sensing via binary stable embeddings of sparse vectors,''
  \emph{IEEE Transactions on Information Theory}, vol.~59, no.~4, pp.
  2082--2102, 2013.

\bibitem{plan2013one}
Y.~Plan and R.~Vershynin, ``One-bit compressed sensing by linear programming,''
  \emph{Communications on Pure and Applied Mathematics}, vol.~66, no.~8, pp.
  1275--1297, 2013.

\bibitem{li2017channel}
Y.~Li, C.~Tao, G.~Seco-Granados, A.~Mezghani, A.~L. Swindlehurst, and L.~Liu,
  ``Channel estimation and performance analysis of one-bit massive mimo
  systems,'' \emph{IEEE Trans. Signal Process}, vol.~65, no.~15, pp.
  4075--4089, 2017.

\bibitem{choi2016near}
J.~Choi, J.~Mo, and R.~W. Heath, ``Near maximum-likelihood detector and channel
  estimator for uplink multiuser massive mimo systems with one-bit adcs,''
  \emph{IEEE Transactions on Communications}, vol.~64, no.~5, pp. 2005--2018,
  2016.

\bibitem{720540}
T.~Han and S.~Amari, ``Statistical inference under multiterminal data
  compression,'' \emph{{IEEE} Transactions on Information Theory}, vol.~44,
  no.~6, pp. 2300--2324, Oct 1998.

\bibitem{cai2020distributed}
T.~T. Cai and H.~Wei, ``Distributed gaussian mean estimation under
  communication constraints: Optimal rates and communication-efficient
  algorithms,'' \emph{arXiv preprint arXiv:2001.08877}, 2020.

\bibitem{gray1998quantization}
R.~Gray and D.~Neuhoff, ``Quantization,'' \emph{{IEEE} Transactions on
  Information Theory}, vol.~44, no.~6, pp. 2325--2383, Oct 1998.

\bibitem{Tsybakov09}
A.~B. Tsybakov, \emph{Introduction to Nonparametric Estimation}.\hskip 1em plus
  0.5em minus 0.4em\relax Springer, 2009.

\bibitem{LeCam86}
L.~{Le Cam}, \emph{Asymptotic Methods in Statistical Decision Theory}.\hskip
  1em plus 0.5em minus 0.4em\relax Springer-Verlag, 1986.

\bibitem{LeCamYa00}
L.~{Le Cam} and G.~L. Yang, \emph{Asymptotics in Statistics: Some Basic
  Concepts}.\hskip 1em plus 0.5em minus 0.4em\relax Springer, 2000.

\bibitem{VanDerVaart98}
A.~W. van~der Vaart, \emph{Asymptotic Statistics}, ser. Cambridge Series in
  Statistical and Probabilistic Mathematics.\hskip 1em plus 0.5em minus
  0.4em\relax Cambridge University Press, 1998.

\bibitem{DavenportPlVaWo15}
M.~A. Davenport, Y.~Plan, E.~van~den Berg, and M.~Wootters, ``One-bit matrix
  completion,'' \emph{Information and Inference}, p. to appear, 2015.

\bibitem{PlanVe13}
Y.~Plan and R.~Vershynin, ``Robust 1-bit compressed sensing and sparse logistic
  regression: A convex programming approach,'' \emph{IEEE Transactions on
  Information Theory}, vol.~59, no.~1, pp. 482--494, 2013.

\bibitem{904560}
W.~Shi, T.~W. Sun, and R.~D. Wesel, ``Quasi-convexity and optimal binary fusion
  for distributed detection with identical sensors in generalized {G}aussian
  noise,'' \emph{{IEEE} Transactions on Information Theory}, vol.~47, no.~1,
  pp. 446--450, Jan 2001.

\bibitem{4244748}
P.~Venkitasubramaniam, L.~Tong, and A.~Swami, ``Quantization for maximin are in
  distributed estimation,'' \emph{IEEE Transactions on Signal Processing},
  vol.~55, no.~7, pp. 3596--3605, July 2007.

\bibitem{6882252}
A.~Vempaty, H.~He, B.~Chen, and P.~K. Varshney, ``On quantizer design for
  distributed bayesian estimation in sensor networks,'' \emph{IEEE Transactions
  on Signal Processing}, vol.~62, no.~20, pp. 5359--5369, Oct 2014.

\bibitem{chen2010performance}
H.~Chen and P.~K. Varshney, ``Performance limit for distributed estimation
  systems with identical one-bit quantizers,'' \emph{IEEE Transactions on
  Signal Processing}, vol.~58, no.~1, pp. 466--471, 2010.

\bibitem{5184907}
------, ``Performance limit for distributed estimation systems with identical
  one-bit quantizers,'' \emph{IEEE Transactions on Signal Processing}, vol.~58,
  no.~1, pp. 466--471, Jan 2010.

\bibitem{berger1996ceo}
T.~Berger, Z.~Zhang, and H.~Viswanathan, ``The {CEO} problem [multiterminal
  source coding],'' \emph{{IEEE} Transactions on Information Theory}, vol.~42,
  no.~3, pp. 887--902, 1996.

\bibitem{viswanathan1997quadratic}
H.~Viswanathan and T.~Berger, ``The quadratic {G}aussian {CEO} problem,''
  \emph{{IEEE} Transactions on Information Theory}, vol.~43, no.~5, pp.
  1549--1559, 1997.

\bibitem{oohama1998rate}
Y.~Oohama, ``The rate-distortion function for the quadratic {G}aussian {CEO}
  problem,'' \emph{{IEEE} Transactions on Information Theory}, vol.~44, no.~3,
  pp. 1057--1070, 1998.

\bibitem{prabhakaran2004rate}
V.~Prabhakaran, D.~Tse, and K.~Ramachandran, ``Rate region of the quadratic
  {G}aussian {CEO} problem,'' in \emph{Information Theory, 2004. ISIT 2004.
  Proceedings. International Symposium on}.\hskip 1em plus 0.5em minus
  0.4em\relax IEEE, 2004, p. 119.

\bibitem{zhang2013information}
Y.~Zhang, J.~Duchi, M.~I. Jordan, and M.~J. Wainwright, ``Information-theoretic
  lower bounds for distributed statistical estimation with communication
  constraints,'' in \emph{Advances in Neural Information Processing Systems},
  2013, pp. 2328--2336.

\bibitem{duchi2014optimality}
J.~C. Duchi, M.~I. Jordan, M.~J. Wainwright, and Y.~Zhang, ``Optimality
  guarantees for distributed statistical estimation,'' \emph{arXiv preprint
  arXiv:1405.0782}, 2014.

\bibitem{GargMaNg14}
A.~Garg, T.~Ma, and H.~L. Nguyen, ``On communication cost of distributed
  statistical estimation and dimensionality,'' in \emph{Advances in Neural
  Information Processing Systems 27}, 2014.

\bibitem{BravermanGaMaNgWo16}
\BIBentryALTinterwordspacing
M.~Braverman, A.~Garg, T.~Ma, H.~L. Nguyen, and D.~P. Woodruff, ``Communication
  lower bounds for statistical estimation problems via a distributed data
  processing inequality,'' in \emph{Proceedings of the Forty-Eighth Annual ACM
  Symposium on the Theory of Computing}, 2016. [Online]. Available:
  \url{https://arxiv.org/abs/1506.07216}
\BIBentrySTDinterwordspacing

\bibitem{DBLP:journals/corr/abs-1802-08417}
\BIBentryALTinterwordspacing
Y.~Han, A.~{\"{O}}zg{\"{u}}r, and T.~Weissman, ``Geometric lower bounds for
  distributed parameter estimation under communication constraints,''
  \emph{CoRR}, vol. abs/1802.08417, 2018. [Online]. Available:
  \url{http://arxiv.org/abs/1802.08417}
\BIBentrySTDinterwordspacing

\bibitem{zhang1988estimation}
Z.~Zhang and T.~Berger, ``Estimation via compressed information,'' \emph{{IEEE}
  Transactions on Information Theory}, vol.~34, no.~2, pp. 198--211, 1988.

\bibitem{han2018distributed}
Y.~Han, P.~Mukherjee, A.~Ozgur, and T.~Weissman, ``Distributed statistical
  estimation of high-dimensional and nonparametric distributions,'' in
  \emph{2018 IEEE International Symposium on Information Theory (ISIT)}.\hskip
  1em plus 0.5em minus 0.4em\relax IEEE, 2018, pp. 506--510.

\bibitem{xu2017information}
A.~Xu and M.~Raginsky, ``Information-theoretic lower bounds on {B}ayes risk in
  decentralized estimation,'' \emph{IEEE Transactions on Information Theory},
  vol.~63, no.~3, pp. 1580--1600, 2017.

\bibitem{Barnes2018}
L.~Barnes, Y.~Han, and A.~Ozgur, ``A geometric characterization of fisher
  information from quantized samples with applications to distributed
  statistical estimation,'' in \emph{2018 56st Annual Allerton Conference on
  Communication, Control, and Computing (Allerton)}, Oct 2018.

\bibitem{52470}
M.~Longo, T.~D. Lookabaugh, and R.~M. Gray, ``Quantization for decentralized
  hypothesis testing under communication constraints,'' \emph{{IEEE}
  Transactions on Information Theory}, vol.~36, no.~2, pp. 241--255, Mar 1990.

\bibitem{tsitsiklis1988decentralized}
J.~N. Tsitsiklis, ``Decentralized detection by a large number of sensors,''
  \emph{Mathematics of Control, Signals, and Systems (MCSS)}, vol.~1, no.~2,
  pp. 167--182, 1988.

\bibitem{5751320}
W.~P. Tay and J.~N. Tsitsiklis, ``The value of feedback for decentralized
  detection in large sensor networks,'' in \emph{International Symposium on
  Wireless and Pervasive Computing}, Feb 2011, pp. 1--6.

\bibitem{ibragimov1956composition}
I.~A. Ibragimov, ``On the composition of unimodal distributions,'' \emph{Theory
  of Probability \& Its Applications}, vol.~1, no.~2, pp. 255--260, 1956.

\bibitem{LehmannCa98}
E.~L. Lehmann and G.~Casella, \emph{Theory of Point Estimation, Second
  Edition}.\hskip 1em plus 0.5em minus 0.4em\relax Springer, 1998.

\bibitem{bagnoli2005log}
M.~Bagnoli and T.~Bergstrom, ``Log-concave probability and its applications,''
  \emph{Economic theory}, vol.~26, no.~2, pp. 445--469, 2005.

\bibitem{Samford1953}
M.~R. Sampford, ``Some inequalities on mill's ratio and related functions,''
  \emph{The Annals of Mathematical Statistics}, vol.~24, no.~1, pp. 130--132,
  1953.

\bibitem{hammersley1950estimating}
J.~Hammersley, ``On estimating restricted parameters,'' \emph{Journal of the
  Royal Statistical Society. Series B (Methodological)}, vol.~12, no.~2, pp.
  192--240, 1950.

\bibitem{chen2004upper}
J.~Chen, X.~Zhang, T.~Berger, and S.~Wicker, ``An upper bound on the sum-rate
  distortion function and its corresponding rate allocation schemes for the
  {CEO} problem,'' \emph{Selected Areas in Communications, IEEE Journal on},
  vol.~22, no.~6, pp. 977--987, Aug 2004.

\bibitem{polyak1992acceleration}
B.~T. Polyak and A.~B. Juditsky, ``Acceleration of stochastic approximation by
  averaging,'' \emph{SIAM Journal on Control and Optimization}, vol.~30, no.~4,
  pp. 838--855, 1992.

\bibitem{shamir2014fundamental}
O.~Shamir, ``Fundamental limits of online and distributed algorithms for
  statistical learning and estimation,'' in \emph{Advances in Neural
  Information Processing Systems}, 2014, pp. 163--171.

\bibitem{braverman2016communication}
M.~Braverman, A.~Garg, T.~Ma, H.~L. Nguyen, and D.~P. Woodruff, ``Communication
  lower bounds for statistical estimation problems via a distributed data
  processing inequality,'' in \emph{Proceedings of the forty-eighth annual ACM
  symposium on Theory of Computing}, 2016, pp. 1011--1020.

\bibitem{han2018geometric}
Y.~Han, A.~{\"O}zg{\"u}r, and T.~Weissman, ``Geometric lower bounds for
  distributed parameter estimation under communication constraints,'' in
  \emph{Conference On Learning Theory}.\hskip 1em plus 0.5em minus 0.4em\relax
  PMLR, 2018, pp. 3163--3188.

\bibitem{barnes2020lower}
L.~P. Barnes, Y.~Han, and A.~Ozgur, ``Lower bounds for learning distributions
  under communication constraints via fisher information,'' \emph{Journal of
  Machine Learning Research}, vol.~21, no. 236, pp. 1--30, 2020.

\bibitem{Bertsekas73}
D.~P. Bertsekas, ``Stochastic optimization problems with nondifferentiable cost
  functionals,'' \emph{Journal of Optimization Theory and Applications},
  vol.~12, no.~2, pp. 218--231, 1973.

\bibitem{gill1995applications}
R.~D. Gill and B.~Y. Levit, ``Applications of the van {T}rees inequality: a
  {B}ayesian {C}ram{\'e}r-{R}ao bound,'' \emph{Bernoulli}, pp. 59--79, 1995.

\bibitem{polyak1990new}
B.~T. Polyak, ``New stochastic approximation type procedures,'' \emph{Automat.
  i Telemekh}, vol.~7, no. 98-107, p.~2, 1990.

\bibitem{beran1995role}
R.~Beran, ``The role of {H}{\'a}jek's convolution theorem in statistical
  theory,'' \emph{Kybernetika}, vol.~31, no.~3, pp. 221--237, 1995.

\bibitem{Dembo16}
\BIBentryALTinterwordspacing
A.~Dembo, ``Lecture notes on probability theory: Stanford statistics 310,''
  2016, accessed October 1, 2016. [Online]. Available:
  \url{http://statweb.stanford.edu/~adembo/stat-310b/lnotes.pdf}
\BIBentrySTDinterwordspacing

\bibitem{robbins1971convergence}
H.~Robbins and D.~Siegmund, ``A convergence theorem for non negative almost
  supermartingales and some applications,'' in \emph{Optimizing methods in
  statistics}.\hskip 1em plus 0.5em minus 0.4em\relax Elsevier, 1971, pp.
  233--257.

\end{thebibliography}

\begin{IEEEbiographynophoto}{Alon Kipnis}
is a Senior Lecturer (Assistant Professor) at the School of Computer Science at Reichman University. Previously, he was a postdoctoral research scholar at the Department of Statistics at Stanford University, advised by David Donoho. He completed his Ph.D. in electrical engineering from Stanford University in 2017. His research is in the areas of mathematical statistics and information theory. 
\end{IEEEbiographynophoto}

\begin{IEEEbiographynophoto}{John Duchi}
is an associate professor of Statistics and Electrical Engineering and (by courtesy) Computer Science at Stanford University. His work spans statistical learning, optimization, information theory, and computation, with a few driving goals. (1) To discover statistical learning procedures that optimally trade between real-world resources---computation, communication, privacy provided to study participants---while maintaining statistical efficiency. (2) To build efficient large-scale optimization methods that address the spectrum of optimization, machine learning, and data analysis problems we face, allowing us to move beyond bespoke solutions to methods that robustly work. (3) To develop tools to assess and guarantee the validity of---and confidence we should have in---machine-learned systems.

He has won several awards and fellowships. His paper awards include the SIAM SIGEST award for "an outstanding paper of general interest" and best papers at the Neural Information Processing Systems conference, the International Conference on Machine Learning, and an INFORMS Applied Probability Society Best Student Paper Award (as advisor). He has also received the Society for Industrial and Applied Mathematics (SIAM) Early Career Prize in Optimization, an Office of Naval Research (ONR) Young Investigator Award, an NSF CAREER award, a Sloan Fellowship in Mathematics, the Okawa Foundation Award, the Association for Computing Machinery (ACM) Doctoral Dissertation Award (honorable mention), and U.C. Berkeley's C.V. Ramamoorthy Distinguished Research Award.

\end{IEEEbiographynophoto}

\end{document}